\documentclass{article}
\usepackage{ijcai23}

\usepackage{times}
\usepackage{soul}
\usepackage{url}
\usepackage[hidelinks]{hyperref}
\usepackage[utf8]{inputenc}
\usepackage[small]{caption}
\usepackage{graphicx}
\usepackage{amsmath}
\usepackage{amsthm}
\usepackage{amssymb}
\usepackage{booktabs}
\usepackage{algorithm}
\usepackage{algorithmic}
\usepackage{multirow}
\usepackage[switch]{lineno}
\usepackage{natbib}
\usepackage[nameinlink]{cleveref}
\usepackage{balance}

\usepackage{tikz}
\usetikzlibrary{graphs}
\usetikzlibrary{decorations.pathreplacing,angles,quotes}
\usetikzlibrary{backgrounds}

\usepackage{thmtools}
\usepackage{thm-restate}
\usepackage{subcaption}

\newtheorem{theorem}{Theorem}
\newtheorem{proposition}{Proposition}
\newtheorem{corollary}{Corollary}
\newtheorem{lemma}{Lemma}

\newtheorem{conjecture}{Conjecture}

\crefname{example}{Ex.}{Exs.}
\crefname{theorem}{Thm.}{Thms.}
\crefname{proposition}{Prop.}{Props.}
\crefname{corollary}{Cor.}{Cors.}
\crefname{lemma}{Lm.}{Lms.}
\crefname{observation}{Obs.}{Obss.}
\crefname{definition}{Def.}{Defs.}
\crefname{conjecture}{Conj.}{Conjs.}

\newcommand*{\scost}{\text{cost}}
\newcommand*{\cost}[1]{\text{cost}_{#1}}

\newcommand*{\maxedge}{\text{maxedge}}
\newcommand*{\co}{\lambda}

\def\hpoa {{\sf PoA}}
\def\hpos {{\sf PoS}}

\def\deg {{\text{deg}}}

\title{Schelling Games with Continuous Types}

\author{
Davide Bilò$^1$
\and
Vittorio Bilò$^2$\and
Michelle Döring$^{3}$\and\\
Pascal Lenzner$^{3}$\and
Louise Molitor$^3$\and
Jonas Schmidt$^3$
\affiliations
$^1$University of L'Aquila, L'Aquila, Italy\\
$^2$Universtiy of Salento, Lecce, Italy\\
$^3$Hasso Plattner Institute, Potsdam, Germany\\
\emails
davide.bilo@univaq.it,
vittorio.bilo@unisalento.it,
\{michelle.doering, pascal.lenzner, louise.molitor\}@hpi.de,
jonas.schmidt@student.hpi.uni-potsdam.de
}

\begin{document}

\maketitle

\begin{abstract}
   In most major cities and urban areas, residents form homogeneous neighborhoods
  along ethnic or socioeconomic lines. This phenomenon is widely known as
  residential segregation and has been studied extensively.
  Fifty years ago,
  Schelling proposed a landmark model that explains residential segregation in an elegant
  agent-based way.
  A recent stream of papers analyzed Schelling's model using game-theoretic approaches. However, all these works considered models with a given number of discrete types modeling different ethnic groups.

  We focus on segregation caused by non-categorical attributes, such
  as household income or position in a political left-right spectrum. For this, we consider agent types that can be represented as real numbers.
  This opens up a great variety of reasonable models and, as a proof of concept, we focus on several natural candidates. In particular, we consider agents that evaluate their location by the average
  type-difference or the maximum type-difference to their neighbors, or by having a certain tolerance range for type-values of neighboring agents.
  We study the existence and computation of equilibria and provide bounds on the Price of Anarchy and Stability. Also, we present simulation results that compare our models and shed light on the obtained equilibria for our variants.
\end{abstract}

\section{Introduction}
"Birds of a feather flock together" is an often used proverb to describe \emph{homophily}~\citep{mcpherson2001birds}, i.e., the phenomenon that homogeneous groups are prevalent in society.
The group members might be similar in terms of, for example, their ethnic group, their socioeconomic status, or their political orientation.
Within a city, such groups typically cluster together, which then leads to segregated neighborhoods, called \emph{residential segregation}~\citep{massey1988dimensions}.

Segregated neighborhoods have a strong impact on the socioeconomic prospects~\citep{massey2019american} and on the health of its inhabitants~\citep{acevedo2003residential,williams2016racial}. This explains why residential segregation is widely studied.
Typical models are agent-based and they assume that the agents are partitioned into a given fixed set of \emph{types}, which can be understood as an ethnic group, a trait, or an affiliation. The landmark model of this kind was proposed by~\citet{schelling69,schelling71} roughly fifty years ago. There, agents of two types are placed on the line or a grid and it is assumed that an agent is content with her current location, if at least a $\tau$-fraction of all neighbors are of her type, for some $\tau\in[0,1]$. Discontent agents try to relocate. As a result, large homogeneous neighborhoods eventually form, even if all the agents are tolerant, i.e., if $\tau\leq \frac12$.

However, real-world agents show more complex behavior than predicted by Schelling's two-type model. For example, people might not care about the ethnic group, but they might compare themselves with their neighbors along non-categorical aspects like age, household income, or position in a political left-right spectrum. Given these more complex preferences, the agents cannot be assumed to simply classify their neighbors into friends and enemies.
\begin{figure*}[ht]
    \centering
    \begin{subfigure}{0.24\textwidth}
        \centering
        \includegraphics[width=0.9\textwidth]{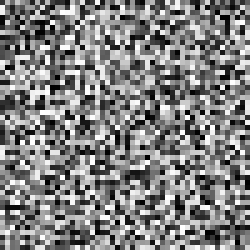}
        \caption{random initial strategy profile\label{fig:initial}}
    \end{subfigure}
    \hfil
    \begin{subfigure}{0.24\textwidth}
        \centering
        \includegraphics[width=0.9\textwidth]{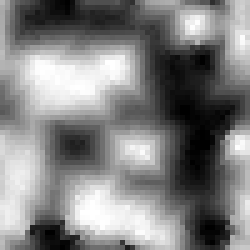}
        \caption{SE for the MDG\label{fig:equimax}}
    \end{subfigure}
    \hfil
    \begin{subfigure}{0.24\textwidth}
        \centering
        \includegraphics[width=0.9\textwidth]{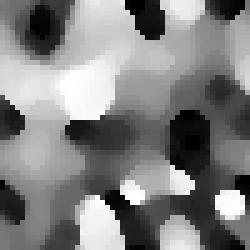}
        \caption{SE for the ADG\label{fig:equiavg}}
    \end{subfigure}
    \hfil
    \begin{subfigure}{0.24\textwidth}
        \centering
        \includegraphics[width=0.9\textwidth]{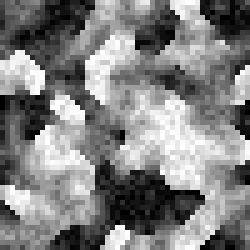}
        \caption{SE for the CG with $\lambda = \frac15$\label{fig:equicut}}
    \end{subfigure}
    \caption{Sample swap equilibria (SE) in different Schelling Games with Continuous Types obtained from the same initial random state \textbf{(a)}, in an 8-regular toroidal grid
    graph of size $50\times 50$\label{fig:equi}. \textbf{(b)}, \textbf{(c)}, and \textbf{(d)} show
    SE for our models reached from \textbf{(a)} via
    improving swaps. Each pixel represents a node with an agent. Colors represent type-values ranging from white (type-value 0) to black (type-value 1).}
\label{fig:swap_comparison}
\end{figure*}
In contrast, the utility of an agent should depend on the respective non-categorical \emph{type-values} of her neighbors. This is in line with recent economics research which reveals that individuals' happiness is relative to a particular peer group. E.g., the reference income hypothesis~\citep{clark1996satisfaction,clark2008relative} states that people compare their income with a reference value, e.g., the mean or median income of their neighborhood~\citep{luttmer2005neighbors,clark2009economic}.

With this paper, we initiate the study of agent-based models for residential segregation that use non-categorical type-values for the agents. This allows for modeling more realistic agent preferences.
Using arbitrary type-values in $[0,1]$ unlocks an entirely new class of game-theoretic models, that we call \emph{Schelling Games with Continuous Types}. As the first steps, we consider three natural behavioral models: agents compare their type-value with the most different or the average type-value in their neighborhood, or they have a tolerance range for the accepted type-value difference. All three variants of the cost function are motivated by plausible real-world behavior. Comparing with the maximum-difference type-value is motivated by considering types as positions in a political left-to-right spectrum, comparing with the average type-value is suggested by the setting where types are household incomes, and the model with a tolerance range is inspired by types being the age of the agents, where agents consider other agents
as similar if they are roughly their age.

As~\autoref{fig:swap_comparison} shows, our models yield very different residential patterns in equilibrium.

\subsection{Model}\label{model}
Given an undirected graph $G = (V, E)$ and a node $v \in V$, we denote the \emph{degree} of $v$ in $G$ as $\deg (v)=|\{u \in V: \{v,u\} \in E\}|$. If every node in $G$ has the same degree $\Delta$, we say that $G$ is $\Delta$-\emph{regular}.
For $x\in \mathbb{N}^+$, let $[x] = \{1,\ldots,x\}$.

A \emph{Schelling Game with Continuous Types} is defined by an undirected connected graph~$G$, $n$ strategic agents and a type function $t: [n] \to [0,1]$ mapping agent $i \in [n]$ to her \emph{type} $t(i) \in [0,1]$. Unless stated otherwise, we assume wlog that $t(i)\leq t(j)$ for $i,j \in [n]$ and $i \leq j$. The \emph{type-distance} $d: [n]^2 \to [0,1]$ between two agents $i$ and $j$ is defined as $d(i,j) = |t(i) - t(j)|$.

An agent's strategy is her location on the graph, i.e., a node of $G$. A \emph{strategy profile} $\sigma$ is an $n$-dimensional vector whose $i$-th entry corresponds to the strategy of the $i$-th agent and where all strategies are pairwise disjoint. Let~$\sigma^{-1}$ be its inverse function, mapping a node $v \in V$ to the agent~$i$ choosing~$v$ as her strategy, with the assumption that $\sigma^{-1}$ is equal to \o{} if~$v$ is empty, i.e., no agent chooses $v$ as her strategy. We denote the set of empty nodes in $\sigma$ as $\mathcal{E}(\sigma) = \{v \in V : \sigma^{-1}(v) = \text{\o{}} \}$. Let $|\mathcal{E}(\sigma)| = e$.~For an agent $i \in [n]$, let the \emph{neighborhood} of $i$ be the set $N_i(\sigma)= \{j \in [n]: \{\sigma(i),\sigma(j)\} \in E\}$ of agents living in the neighborhood of~$\sigma(i)$ in $G$. If $N_i(\sigma) = \emptyset$, we say that $i$ is \emph{isolated}.

In \emph{Swap Schelling Games with Continuous Types}, every node is occupied by exactly one agent, so $n=|V|$.
Agents can change their strategies only by swapping their location with another agent. A \emph{swap} by agents $i$ and $j$ in $\sigma$ yields a new strategy profile $\sigma_{ij}$, which is identical to~$\sigma$ with exchanged $i$-th and $j$-th entries.
As agents are rational, we only consider \emph{profitable swaps} that strictly decrease the individual cost of both involved agents. A strategy profile is a \emph{swap equilibrium (SE)} if it does not admit any
profitable swaps.

In \emph{Jump Schelling Games with Continuous Types}, empty nodes exist, i.e., $n<|V|$, and for every strategy profile~$\sigma$ we have $e=|V|-n$. An agent can change her strategy by jumping to any empty node. Consider agent $i$ on node~$\sigma(i)$. A {\em jump} of agent $i$ to an empty node $v \in \mathcal{E}(\sigma)$ yields a new strategy profile $\sigma_i$, which is identical to~$\sigma$ with the $i$-th entry changed to $v$. Agents only perform \emph{profitable jumps} that strictly decrease their cost. A strategy profile is a \emph{jump equilibrium (JE)} if it does not admit any profitable jumps.

We consider the following three cost models. In \emph{Average Type-Distance Games (ADGs)}, the cost of agent~$i$ in~$\sigma$ is defined as the average distance towards her neighbors, i.e., $\cost{i}(\sigma) = \frac{\sum_{j \in N_i(\sigma)} d(i,j)}{|N_i(\sigma)|}$. In \emph{Maximum Type-Distance Games (MDGs)}, the cost of agent $i$ in $\sigma$ is defined as the maximum distance towards her neighbors, i.e., $\cost{i}(\sigma) = \max_{j \in N_i(\sigma)} d(i,j)$. In \emph{Cutoff Games (CGs)}, given a cutoff parameter $\co\in [0,1]$, let $N^+_i(\sigma)=\{j\in N_i(\sigma):d(i,j)\leq\co\}$ be the set of {\em friends} of agent $i$ in $\sigma$ and $N^-_i(\sigma)=N_i(\sigma)\setminus N^+_i(\sigma)$ be the set of {\em enemies}. The cost $i$ in $\sigma$ is the fraction of enemies in the neighborhood of $i$, i.e., $\cost{i}(\sigma)= \frac{|N_i^-(\sigma)|}{|N_i(\sigma)|}$. The model of CGs is closer to the original Schelling model, i.e., neighbors whose type difference is within the cutoff are considered as \emph{friends}; however, in contrast to previous models, friendship is not transitive.

For all cost models, we consider two possible variants depending on how we define the cost of an isolated agent. Under the {\em unhappy-in-isolation} (UIS) variant, this cost is set to~$1$; under the {\em happy-in-isolation} (HIS) variant, it is set to~$0$. Since we are considering connected graphs, an agent can never be isolated in swap games, so the two variants create a different model only for jump games.
In summary, we obtain nine different games that we denote as X-Y-Z, where X$\in$\{J,S\} stands for the deviation model, either jump (J) or swap (S), Z$\in$\{ADG,MDG,CG\} stands for cost model and, when X = J, Y$\in$\{UIS,HIS\} states which cost is paid in isolation.

We measure the quality of a strategy profile $\sigma$ by its \emph{social cost} $\scost(\sigma) = \sum_{i \in [n]} \cost{i}(\sigma)$ and denote by $\sigma^*$ a {\em social optimum}, i.e., a strategy profile minimizing the social cost\footnote{The social cost can be considered as a segregation measure, since a low social cost in our models means that many agent neighborhoods are very homogeneous, i.e., are strongly segregated.}. The quality of equilibria is measured by the price of anarchy ($\hpoa$) and the price of stability ($\hpos$). The $\hpoa$ of a game $\cal G$ is obtained by comparing the equilibrium with the largest social cost with the social optimum, while the $\hpos$ refers to the equilibrium with the lowest social cost. The $\hpoa$ (resp. $\hpos$) of a class of games $\cal C$ is obtained by taking the worst-case $\hpoa$ (resp. $\hpos$) over all games in the class. Formally, $\hpoa({\cal C})=\sup_{{\cal G}\in{\cal C}}\hpoa({\cal G})$ and $\hpos({\cal C})=\sup_{{\cal G}\in{\cal C}}\hpos({\cal G})$.

A game has the finite improvement property (FIP) if any sequence of improving moves must be finite. For showing this, we employ \emph{ordinal potential functions}, since the FIP is equivalent to the existence of an ordinal potential function~\citep{MS96}.

\subsection{Related Work}
Given the breadth of the research on residential segregation, we focus on recent work from the Artificial Intelligence and Algorithmic Game Theory communities.

Schelling's seminal residential segregation model was formulated as a strategic game by \citet{CLM18}. In their model, agents of two types have a threshold-based utility function and an agent gets maximum utility if for this agent the fraction of same-type neighbors is at least $\tau$. Later, \citet{E+19} extended this model to more than two types and showed that the convergence behavior of improving response dynamics strongly depends on $\tau$. \citet{A+19} focused on the case with $\tau=1$ and proved that equilibrium existence on trees is not guaranteed and that computing socially optimal strategy profiles or equilibria with high social welfare is NP-hard. Also, the authors introduced a new welfare measure that counts the number of agents that have an other-type neighbor, called the degree of integration. For $\tau=1$ also the influence of the underlying graph and of locality was studied~\citep{bilo20} and welfare guarantees have been investigated~\citep{bullinger21}.

A variant where the agent itself is counted in the fraction of same-type neighbors has been introduced in~\citep{KKV21}. Moreover, recently also agents with non-monotone utility functions, in particular, with single-peaked utilities, have been considered in~\citep{BBLM22}.

Closest to our work is another very recent variant, called Tolerance Schelling Games, introduced by \citet{KKV22}.
In this model agents have a discrete type and all $k$ types are ordered according to a given total ordering $\succ$, i.e., $T_1 \succ T_2 \succ \dots \succ T_k$. Agents have tolerance values to agents of other types depending on the number of types in between the two in the given ordering. Specifically, the model uses a tolerance vector $\mathbf{t} = (t_0, \dots, t_{k-1})$ and the tolerance between agents of type $T_i$ and type $T_j$ is equal to $t_{|i-j|}$. In their work, they specifically analyze balanced tolerance Schelling games in which every type has the same number of agents and only consider the jump variant of the game.
The authors show that for every tolerance vector with $t_1 < 1$ there are graphs that do not admit equilibria. Furthermore, they look at \emph{$\alpha$-binary} Tolerance Schelling Games where agents tolerate all other agents of types with at most $\alpha-1$ other types in between in the ordering. For specific values of $\alpha$ and $k$ this game admits at least one equilibrium on trees and grid graphs and they provide algorithms to find such states. Also, they prove high tight asymptotic bounds on the $\hpoa$ and the $\hpos$.

\subsection{Our Contribution} We introduce very general strategic residential segregation models with the decisive new feature that non-categorical types are possible. This allows for modeling more complex and arguably also more realistic agent behavior. Moreover, the power of our models can be seen by noting that they generalize several existing variants. For example, the $k$-type model by \citet{A+19} with $k=2$ can be captured by both the ADG and the CG, by setting one type to value 0 and the other to value 1 (and $\lambda<1$). Also, the $k$-type model by~\cite{E+19} for $\tau=1$ can be modeled via the CG with suitable type-values and low enough $\lambda$. Moreover, also the $\alpha$-binary Tolerance Schelling Game by~\citet{KKV22} is captured by the CG, with equally spaced type-values and a suitably chosen cutoff $\lambda$\footnote{In the other direction, the J-UIS-ADG can be represented as a Tolerance Schelling Game by using enough types, with the tolerance of two types $x,y\in [0,1]$ being $1-|r-s|$. However, irrational type-values cannot be translated.}.

Besides generalizing several known models, our results go beyond what was known for the special cases. In particular, we demonstrate with the MDG, that our model allows for drastically different games that behave very differently, compared to previously considered variants. For the S-MDG, not only do equilibria always exist, independently of the underlying graph, but we are also able to construct these states very efficiently. The same holds for the J-HIS-MDG, the ADG, and the CG on specific graph classes. Also, the HIS-assumption has not been studied before. See~\Cref{tbl:result_overview} for an overview over our equilibrium existence results.

Besides equilibrium existence, we also studied the computational complexity of computing the social optimum and of finding states that minimize the maximum type-difference of neighbors. Both problems are NP-hard, see~\Cref{apx:complexity}. Moreover, we provide extensive results on the Price of Anarchy and the Price of Stability. In essence, the PoA is extremely high for all variants but the PoS is very low on paths or regular graphs. See~\Cref{sec:quality} for a high-level overview and~\Cref{apx:quality} for all the details.
Finally, to shed light on the equilibrium structure and properties of our game variants, we present simulation experiments in~\Cref{sec:experiments}. We find that the MDG and the ADG yield strongly segregated equilibria, while the CG produces more integrated equilibria.

All omitted details can be found in the appendix.

\begin{table*}[h]
    \centering
    \footnotesize
    \resizebox{\textwidth}{!}{\begin{tabular}{@{}lrrrrrrrrr}
        \toprule
        \multirow{2}{4em}{Equilibrium existence}& \multicolumn{3}{c}{MDG} & \multicolumn{3}{c}{ADG} & \multicolumn{3}{c}{CG} \\
        \cmidrule(l){2-4} \cmidrule(l){5-7}\cmidrule(l){8-10}
        & path & regular & general & path & regular & general & path & regular & general \\
        \midrule
        Swap
        &  $\checkmark$$\checkmark$ (\cref{thm:Swap_MDG_Comp})& $\checkmark$$\checkmark$ (\cref{thm:Swap_MDG_Comp}) & $\checkmark$$\checkmark$ (\cref{thm:Swap_MDG_Comp})
        & $\checkmark$$\checkmark$ (\cref{PoS-AVG-path}) & $\checkmark$ (\cref{thm:Swap_ADG_Ex}) & $\times$ (\cref{from_Oxford})
        & $\checkmark$$\checkmark$ (\cref{thm:Swap_CG_Conv_almost_regular}) & $\checkmark$$\checkmark$ (\cref{thm:Swap_CG_Conv_almost_regular}) & $\times$ (\cref{from_Oxford})\\
        Jump
        &  & $(\times)$ (\cref{thm:Jump_MDG_Ex_reg}) & $(\times)$ (\cref{thm:Jump_MDG_Ex_reg})
        &  & $\times$ (\cref{thm:Jump_MDG_Ex_reg}) & $\times$ (\cref{from_Oxford})
        &  &  & $\times$ (\cref{from_Oxford}) \\
        Jump HIS
        &  $\checkmark$$\checkmark$(\cref{jump-path}) & ($\checkmark$)$\checkmark$ (\cref{thm:Jump_HISMDG_Ex})& ($\checkmark$)$\checkmark$(\cref{thm:Jump_HISMDG_Ex})
        & $\checkmark$$\checkmark$ (\cref{jump-path-ext}) &  &
        & $\checkmark$$\checkmark$ (\cref{jump-path-ext}) &  &  \\
        \bottomrule\\
    \end{tabular}}
\caption{Results overview. The symbol "$\checkmark$" means that equilibria are guaranteed to exist, the symbol "$\checkmark\checkmark$" means an equilibrium always exists and can be computed efficiently, the symbol "$(\checkmark)\checkmark$" means that an equilibrium always exists and it can be computed efficiently in some specific cases,  the symbol "$\times$" means that equilibria are not guaranteed to exist, and the symbol $(\times)$ means that, although equilibria are not guaranteed to exist in general, there are some families of instances for which they do exist. This is the case for the J-UIS-MDG on general graphs, which has the FIP when the number $e$ of empty nodes is strictly smaller than the minimum degree of the graph (\cref{thm:Jump_HISMDG_Ex}).
A JE for the J-HIS-MDG can be computed in polynomial time in any graph that contains $K_{2,e}$ as a subgraph (\cref{thm:Jump_HISMDG_Comp_path}). As a byproduct, for the J-HIS-MDG, we can compute a JE in polynomial time when $e\in\{1,2\}$ (\cref{cor:J_HIS_MDG_one_empty_node} and~\Cref{thm:Jump_HISMDG_Comp}) or when the graph is dense (\cref{cor:J_HIS_MDG_dense_graphs}).
}\label{tbl:result_overview}
\end{table*}

\section{Existence and Construction of Equilibria}\label{sec:existence}

First of all, we observe that whenever the type function $t$ is such that $t(0)=1$, $t(1)=1$ and $t(i)\in\{0,1\}$ for each $i\in [n]$, ADGs and CGs boil down to classical Swap or Jump Schelling games with two types considered in \cite{A+19} for which non-existence of equilibria is known in general graphs. Hence, we immediately get the following result.
\begin{proposition}\label{from_Oxford}
S-ADGs, S-CGs, J-UIS-ADGs, and J-UIS-CGs may not have equilibria when played on general graphs.
\end{proposition}

\subsection{Swap Games}
Despite the negative result from \autoref{from_Oxford}, we show that a SE always exists when considering S-MDGs and it can even be efficiently computed.

\subsubsection{Maximum Type-Distance Game}
 The following lemma states that an improving swap never hurts those agents who end up paying a sufficiently large cost after the swap.
\begin{lemma}\label{lma:maximum}
Let agents $i, j \in [n]$ perform a profitable swap in a S-MDG and let $k$ be an agent with $\cost{k}(\sigma_{ij}) \ge\max\{\cost{i}(\sigma),\cost{j}(\sigma)\}$. Then, $\cost{k}(\sigma) \ge\cost{k}(\sigma_{ij}).$
\end{lemma}

\begin{proof}
Assume the cost of agent $k$ increases. Since the cost of $k$ changes, agent $i$ or~$j$ must belong to $N_k(\sigma_{ij})$ and must be crucial for
the new cost of $k$. Assume wlog that agent~$i$ is responsible for the increased cost, so $\cost{k}(\sigma_{ij}) = d(k,i)$. Since $i\in N_k(\sigma_{ij})$, we get $\cost{i}(\sigma_{ij}) \ge d(k,i)$. Thus, we obtain $\cost{i}(\sigma_{ij}) \ge \cost{k}(\sigma_{ij}) \ge\cost{i}(\sigma)$ which contradicts the assumption that agent $i$ performs a profitable swap.
\end{proof}

We show that SE exist for S-MDGs since the FIP holds.

\begin{proposition} \label{thm:Swap_MDG_Ex}
The S-MDG has the FIP.
\end{proposition}
\begin{proof}
By~\autoref{lma:maximum}, the lexicographical order of the $n$-dimensional vector which sorts the agents' costs in a non-increasing way decreases after each improving swap.
\end{proof}

  \autoref{thm:Swap_MDG_Ex} implies that any sequence of improving swaps converges to a SE. However, as there is no guarantee of polynomial-time convergence, this result does not yield a polynomial-time algorithm for computing a SE. The following result shows how to efficiently construct a SE. To this end, we need some additional notation.
  Given a strategy profile $\sigma$ and an agent $i$, we define the {\em leftmost neighbor} of $i$ the agent $l_{\sigma}(i)=\min_{j\in N_i(\sigma)}j$ of smallest index among the neighbors of $i$ in $\sigma$ and the {\em rightmost neighbor} of~$i$ the agent $r_{\sigma}(i)=\max_{j\in N_i(\sigma)}j$ of largest index among the neighbors of $i$ in $\sigma$. As $G$ is connected, $\sigma$ is injective, and there are $|V|$ agents, the leftmost and rightmost neighbors always exist. Moreover, by definition, we have $\cost{i}(\sigma)=\max\{d(i,l_{\sigma}(i)),d(r_{\sigma}(i),i)\}$.

\begin{theorem} \label{thm:Swap_MDG_Comp}
A SE for the S-MDG can be computed in $O(|E|)$.
\end{theorem}
\begin{proof}
For a given graph $G$ defining an S-MDG, let $\tau$ be an arbitrary node of $G$ and let $T$ be a breadth-first search (BFS) tree of $G$ rooted at $\tau$. This tree exists as $G$ is connected. Number the nodes of $T$ from $1$ to $n$ in a natural way, that is, in increasing order of levels and proceedings from left to right within the same level. Let $\sigma$ be the strategy profile such that, for every $i\in [n]$, the $i$-th agent is assigned to the $i$-th node.

A BFS tree of a connected graph can be computed in $O(|E|)$. So, $\sigma$ can be computed in $O(|E|+n)=O(|E|)$. For a node $u$, let~$p(u)$ denote $u$'s parent in $T$. It is well known that a BFS tree satisfies the following properties: (1) every non-tree edge of $G$ only connects nodes of the same level or nodes of two consecutive levels; (2) there cannot be a non-tree edge connecting a node $u$ to a node at the left of~$p(u)$.
Exploiting these claims, we immediately obtain that~$\sigma$ satisfies the following property (p1):
 for every agent $i$ not assigned to $\tau$, $l_{\sigma}(i)=\sigma^{-1}(p(\sigma(i)))$, i.e., the leftmost agent of~$i$ is the agent assigned to the parent of the node to which $i$ is assigned; so, for every $1<i<j$, we have $l_{\sigma}(i)\leq l_{\sigma}(j)$.

We now show that $\sigma$ is a SE. To this end, we prove that for any $i\in [n]$, agent $i$ has no incentive in swapping her position with agent $j$, for every $j>i$.

Let us first consider agent $1$ which is assigned to~$\tau$. In this case, we have $\cost1(\sigma)=d(r_{\sigma}(1),1)=d(\deg(\tau)+1,1)$. Agent $1$ will be interested in swapping with $j>1$ only if~$j$ is surrounded by agents of index smaller than $\deg(\tau)+1$, which holds only if $\sigma(j)$ is adjacent to nodes belonging to either level $0$ or level $1$ of $T$, i.e., only if $\sigma(j)$ is either a child or a nephew of $\tau$ in $T$, with no edges (either tree and non-tree ones) towards nodes of level $\ell\geq 2$. If $\sigma(j)$ is a child of $\tau$, the leftmost neighbor of $j$ does not change after the swap (it is $1$), while, as $\sigma(j)$ is not adjacent to nodes of a level larger than~$1$, the rightmost neighbor of $j$ does not decrease, which prevents~$j$ from swapping. So, $\sigma(j)$ can only be a nephew of~$\tau$. In this case, as $\sigma(j)$ can only be adjacent to nodes of level~$1$, the leftmost neighbor of $j$ does not increase after the swap, while the rightmost one does not decrease, again preventing~$j$ from swapping. So, agent $1$ is not interested in swapping with any other agent.

Now, consider a swap between agent $i>1$ and agent $j>i$. As, by property (p1), $l_{\sigma}(j)\geq l_{\sigma}(i)$, in order for $j$ to be willing to swap, $r_{\sigma}(j)$ must be such that $\cost j(\sigma)=d(r_{\sigma}(j),j)$ and $d(r_{\sigma}(j),j)>d(j,l_{\sigma}(i))$. Now, as $d(r_{\sigma}(j),j)>d(j,l_{\sigma}(i))$ implies that $d(r_{\sigma}(j),i)>d(i,l_{\sigma}(i))$, in order for~$i$ to be willing to swap, $r_{\sigma}(i)$ must be such that $\cost i(\sigma)=d(r_{\sigma}(i),i)$ and $r_{\sigma}(i)>r_{\sigma}(j)$. But this implies that, after the swap, $j$ pays at least $d(r_{\sigma}(i),j)>d(r_{\sigma}(j),j)=\cost j(\sigma)$ preventing $j$ from swapping. So, agent $i$ is not interested in swapping with any other subsequent agent.
\end{proof}

  \subsubsection{Average Type-Distance Games and Cutoff Games}
As stated by~\autoref{from_Oxford}, equilibria are not guaranteed to exist for these games in general topologies. For games played on $\Delta$-regular graphs, however, convergence to  equilibria is recovered as shown in the following theorems.

 \begin{restatable}{proposition}{propthree}
 \label{thm:Swap_ADG_Ex}
			The S-ADG on $\Delta$-regular graphs has the FIP.
        \end{restatable}
		\begin{proof}[Proof Sketch]
			We show this proposition using a more general version of the ordinal potential function
			introduced by \cite{CLM18}. Let $G$ be a $\Delta$-regular graph. We define our ordinal potential function as the sum of distances over all edges in the graph according to a given strategy profile
			$\sigma$: \[\Phi(\sigma) =
			\sum_{\{u, v\} \in E} d(\sigma^{-1}(u), \sigma^{-1}(v)) = 2 \cdot
			\scost(\sigma).\] This is an ordinal potential function.
		\end{proof}

A graph $G$ is \emph{almost $\Delta$-regular} if every of its nodes has degree in $\{\Delta,\Delta+1\}$. We show that polynomial-time convergence for S-CG is also guaranteed on almost $\Delta$-regular graphs. Note that paths are almost $2$-regular graphs.

\begin{restatable}{proposition}{propfour}
\label{thm:Swap_CG_Ex_almostreg}
The S-CG on almost $\Delta$-regular graphs has the FIP.
\end{restatable}
\begin{proof}[Proof Sketch]
Let $G$ be an almost $\Delta$-regular graph.
Consider the function $\Phi(\sigma) = \lvert \{u,v\}\in E \colon \sigma^{-1}(u)\in N_{\sigma^{-1}(v)}^+(\sigma)\}\lvert$ which counts the number of \textit{monochromatic} edges, i.e., all edges whose endpoints are occupied by agents who are friends.
We prove the claim by showing that $\Phi$ is an ordinal potential function. See \Cref{apx:existence} for details.
\end{proof}

  \begin{theorem} \label{thm:Swap_CG_Conv_almost_regular}
			The convergence time for the S-CG on almost $\Delta$-regular graphs is $O(\Delta n)$.
		\end{theorem}

		\begin{proof}
		  In the proof of~\autoref{thm:Swap_CG_Ex_almostreg} we showed that the function $\Phi(\sigma) = \lvert \{u,v\}\in E \colon \sigma^{-1}(u)\in N_{\sigma^{-1}(v)}^+(\sigma)\}\lvert$ is an ordinal potential function that
			maps strategy profiles to integer values from $0$ to at most $(\Delta + 1)n$. As the potential function increases with
			every improving swap, there can be at most $O(\Delta n)$
			improving swaps to any equilibrium.
		\end{proof}

\subsection{Jump Games}
For jump games, stability is harder to achieve and we can prove positive results mostly only under the HIS variant.

  \subsubsection{Maximum Type-Distance Games}

  For these games, the happy-in-isolation assumption makes a huge difference. Using it enables the existence of JE.

  \begin{proposition} \label{thm:Jump_HISMDG_Ex}
			The J-HIS-MDG and the J-UIS-MDG when $e$ is smaller than the minimum degree of $G$ have the FIP.
		\end{proposition}
		\begin{proof}
			The same ordinal potential function as in~\autoref{thm:Swap_MDG_Ex} works. Note that if an
			agent $i \in [n]$ performs an improving jump from $\sigma$,
			every former neighbor of $i$ does not increase her cost unless we are in the UIS variant and the agent becomes isolated. But this can never happen under the claimed premises. By the arguments
			presented in the proof of~\autoref{lma:maximum}, no new neighbor of $i$ can increase her cost
			to be at least $\cost{i}(\sigma)$. Because agent $i$'s cost
			strictly decreases, the potential function lexicographically decreases.
		\end{proof}

  We improve the negative result from~\autoref{from_Oxford} and also complement the above theorem by showing that JE for J-UIS-MDGs may not exist even when considering $\Delta$-regular graphs and the number of empty nodes is equal to $\Delta$.

\begin{proposition} \label{thm:Jump_MDG_Ex_reg}
J-UIS-MDGs and J-UIS-ADGs on $\Delta$-regular graphs may not have JE when $e\geq\Delta$.
\end{proposition}
\begin{proof}
The following instance works for both J-UIS-MDGs and J-UIS-ADGs. Consider a $5$-node ring and three agents of types $t_1=0$, $t_2=1/3$ and $t_3=1$, respectively; so, we have $\Delta=e=2$. Consider a strategy profile $\sigma$. If agent $3$ is isolated in~$\sigma$, then she can jump on the empty spot adjacent to agent $2$ and improve her cost. So, agent $3$ cannot be isolated in $\sigma$. Let $i$ be an agent adjacent to agent $3$ in $\sigma$. This agent can jump to an empty spot where she is the only neighbor of agent $j\notin\{i,3\}$, thus decreasing her cost. So, no JE exists.
\end{proof}

For the rest of the section, we shall focus on J-\textsc{HIS}-MDGs. We first observe that a JE can be efficiently computed when the input graph satisfies a topological property.

\begin{theorem} \label{thm:Jump_HISMDG_Comp_path}
A JE for the J-HIS-MDG can be computed in $O(|V|^{k+1})$ time in any graph having $K_{2,e}$ as a subgraph, where $k=\min\{2,e\}$.
\end{theorem}
\begin{proof}
Let $G$ be a graph having $K_{2,e}$ as a subgraph. So, there are two nodes $u$ and $v$ that are both adjacent to all nodes of a set of nodes $S$ such that $|S|=e$ and $S\cap\{u,v\}=\emptyset$. Let~$\sigma$ be any strategy profile assigning agent $1$ to $u$, agent~$n$ to~$v$ and leaving empty all nodes in $S$. Observe that $\sigma$ is a~JE. This is because by jumping to any node in $S$, an agent would be adjacent to both agents $1$ and $n$ paying her largest possible cost. Nodes $u$ and $v$ can be discovered in $O(|V|^{k+1})$ time by guessing all $k$-tuples of nodes in $V$ and checking in $O(|V|)$ whether their neighborhoods satisfy the required property.
\end{proof}

This yields an efficient algorithm for the following cases.
\begin{corollary}\label{cor:J_HIS_MDG_one_empty_node}
A JE for the J-HIS-MDG can be computed in $O(|V|^2)$ time when there is only one empty node.
\end{corollary}
\begin{proof}
Since every connected graph with at least $3$ nodes admits a node of degree at least $2$, i.e., it contains $K_{2,1}$, the claim follows by~\autoref{thm:Jump_HISMDG_Comp_path}.
\end{proof}

\begin{corollary}\label{cor:J_HIS_MDG_dense_graphs}
A JE for the J-HIS-MDG on graphs with $|V|$ nodes and $\omega(|V|^{3/2})$ edges can be computed in $O(|V|^3)$ time.
\end{corollary}
\begin{proof}
\citet{Kovari1954} proved that the densest graph on~$|V|$ nodes that does not contain $K_{s,t}$ as a subgraph has size $O(|V|^{2-1/s})$. So, every graph with $\omega(|V|^{3/2})$ edges contains $K_{2,e}$ as a subgraph. The claim follows by~\autoref{thm:Jump_HISMDG_Comp_path}.
\end{proof}

With additional work, we can obtain efficient computation of a JE also for the case of $e=2$.

\begin{theorem} \label{thm:Jump_HISMDG_Comp}
A JE for the J-HIS-MDG can be computed in $O(|V|^3)$ when there are only two empty nodes.
\end{theorem}
\begin{proof}
Let $G$ be the graph defining the game. By~\autoref{thm:Jump_HISMDG_Comp_path}, the claim follows when $G$ contains the cycle $C_4$. We additionally show that a JE can be computed in $O(|V|^3)$ time if~$G$ admits a path of five nodes.
One can easily find a path of $5$ nodes in $G$ or say that such a path does not exist in $O(|V|^3)$ time. Indeed, for every two distinct nodes $u$ and $v$ of $G$ each having at least two neighbors not in $\{u,v\}$, there is a path of~$5$ nodes using $u$ and $v$ as its second and fourth node, respectively, if and only if (i) $u$ and $v$ have a neighbor in common (the third node of the path), (ii) either one of the two nodes has at least 3 neighbors not in $\{u,v\}$ or the two nodes share only one common neighbor that is not in $\{u,v\}$.

Let $\pi=\langle u_1,\ldots,u_5\rangle$ be a five-node path in $G$. We assume wlog that nodes are numbered in such a way that the number of neighbors of $u_1$ not in $\pi$ is not larger than the number of neighbors of $u_5$ not in $\pi$. Define $t_M=(t(1)+t(n))/2$ as the middle point between the two extreme types. Let $A^-=\{i\in [n]:t_i\leq t_M\}$ and $A^+=\{i\in [n]:t_i\geq t_M\}$. We construct a JE, depending on which of these sets is larger.

Assume that $|A^+|\geq |A^-|$. Let $\sigma$ be the strategy profile defined as follows: agent $n$ is assigned to $u_1$, agent $1$ to $u_3$ and agent $n-1$ to $u_5$; $u_2$ and $u_4$ are left empty; all remaining neighbors of $u_1$ are filled with agents from $|A^+|$ (this is always possible as $|A^+|\geq |A^-|$ and $u_1$ has less neighbors than $u_5$ outside $\pi$); the remaining agents are randomly placed. We claim that $\sigma$ is a JE. Clearly, no agent is interested in jumping to $u_2$ as this would yield the largest possible cost. So, an agent $i$ may only be interested in jumping to $u_4$. In this case,~$i$ would be adjacent to agents $1$ and $n-1$. So, to have an improving jump, it must be $t_{n-1}<t_n$, $i$ must be adjacent to $u_1$ in $\sigma$, and $i\in A^-$, but this never happens in $\sigma$.

If $|A^+|<|A^-|$, it suffices swapping agents $1$ and $n$, assigning agent $2$ to $u_3$ and filling all remaining neighbors of~$u_1$ with agents from $|A^-|$.

If $G$ does not contain $C_4$, then either {\em (i)} it is a tree or {\em (ii)} it contains only cycles of length three.

If $G$ is a tree, as it cannot have a five-node path, either $G$ is a star or $G$ is a star with one of its leaves being, in turn, the center of a star. In both cases, the assignment in which all agents are isolated is a JE.

If $G$ contains two disjoint cycles of length three, as $G$ is connected, these cycles need to be connected, thus creating a five-node path. If the two cycles are not disjoint, then they may have one or two nodes in common. In the first case, a five-node path arises, in the second one, $C_4$ arises. So, $G$ has exactly one cycle of length three. If two nodes of this cycle have neighbors outside the cycle, then a five-node path arises. So, $G$ can only be a star in which two leaves are adjacent. Leaving empty the center of the star and one of these leaves results in a JE as all agents are isolated.
\end{proof}
We show how to efficiently compute equilibria for paths.
\begin{theorem}\label{jump-path}
A JE for the J-HIS-MDG on paths can be computed in $O(n\log n)$ time.
\end{theorem}
\begin{proof}
Clearly, if the number $e$ of empty nodes satisfies $e\geq n-1$, a strategy profile of social cost equal to  zero can be obtained by making every agent isolated. Such a profile is trivially a JE and a social optimum. So, assume that $e<n-1$.

Let $S$ be a set of $e$ pairs of consecutive agents yielding the largest intervals occurring between the types of two consecutive agents, and let $\sigma$ be the strategy profile obtained by placing the agents in increasing order along the path and leaving an empty spot between any two agents in $S$. Clearly, $\sigma$ can be computed in $O(n\log n)$ time. We claim that $\sigma$ is a JE.

To see this, let us consider an agent $i$ who is willing to jump to an empty node~$u$. By construction, both nodes adjacent to $u$ are occupied by two consecutive agents. Let $k$ and $k+1$ be these agents. As $i$ improves by jumping, $i$ is not isolated in $\sigma$, and so we have $\cost i(\sigma)=|t(i)-t(j)|$, for some $j\in\{i-1,i+1\}$. Assume that $j=i-1$. Since $i$ improves by jumping to $u$, the types of agents $k$ and $k+1$ are closer than $t(j)=t(i-1)$ to $t(i)$. So, it must be $\max\{k,k+1\}\geq i+1$ (observe that it may be $k=i$). We derive $\cost i(\sigma)=t(i)-t(i-1)>t(k+1)-t(i)\geq t(k+1)-t(k)$, $(k,k+1)\in S$ and $(i-1,i)\notin S$, which implies $t(i)-t(i-1)\leq t(k+1)-t(k)$, a contradiction. The case $j=i+1$ is analogous.
\end{proof}

\subsubsection{Average Type-Distance Games and Cutoff Games}
We show that the JE for J-HIS-MDGs on paths returned by the algorithm defined in the proof of~\autoref{jump-path} remains stable also when considering J-HIS-ADGs and J-HIS-CGs.

\begin{restatable}{theorem}{theoremjumppathext}\label{jump-path-ext}
A JE for both J-HIS-ADGs and J-HIS-CGs on paths can be computed in $O(n\log n)$ time.
\end{restatable}

\section{Quality of Equilibria}\label{sec:quality}
In this section, we provide an overview of our results for the $\hpoa$ and $\hpos$ of the considered games. For the former, in particular, we were able to provide a full characterization, while, for the latter, we give results for games played on specific topologies. More specific details can be found in~\Cref{apx:quality}.

\paragraph{Price of Anarchy.}
Under a worst-case view, we can prove that there can be equilibria of positive social cost while a social optimum with social cost zero exists, yielding an unbounded $\hpoa$. The only exceptions are S-MDGs and S-ADGs having a $\hpoa$ in $\Theta(n)$ and $\Theta(n\Delta)$, respectively, where $\Delta$ is the maximum degree of the underlying graph.

\paragraph{Price of Stability.}
For characterizing the $\hpos$, usually, either the FIP or the existence of algorithms computing equilibria with provable approximation guarantees are required. As we have seen in~\Cref{sec:existence}, this may be either impossible or require quite an effort; nevertheless, these difficulties are common also in previous models of Schelling games.
For games played on paths, we derive a bound of $1$ in S-ADGs and an upper bound of $2$ in S-MDGs, S-CGs and J-HIS-MDGs. On regular graphs, a bound of $1$ holds for both S-ADGs and S-CGs. Finally, for games played on unrestricted topologies, we show that the $\hpos$ is in $\Theta(n)$ for both S-MDGs and J-HIS-MDGs and even unbounded for J-UIS-MDGs.

\section{Simulation Experiments}\label{sec:experiments}

We present simulation results to highlight some properties of the obtained equilibria for our model variants.

For our simulations, we consider $8$-regular toroidal grid graphs of size $50 \times 50$ with a
total of $2500$ nodes and $10000$ edges as a residential area. Agent types and starting locations are chosen uniformly at random and are the same if we compare different variants. For simulations of the jump versions, we use 2\% uniformly random chosen empty nodes.
From the starting location, random improving moves are chosen until an equilibrium is reached. The shown equilibria are representative for all runs. For the jump game, we used the UIS variant, thus, we had no convergence guarantee. However, the simulations always found a JE.

\paragraph{Visualizations of Equilibria.} \autoref{fig:swap_comparison} depicts representative sample equilibria for our variants of the swap game. The differences in the equilibria are remarkable. While the SE for the MDG looks very smooth, the SE of the other variants show hard color borders, i.e., they have neighboring agents with large type-difference. This indicates that the segregation strength in the MDG seems to be higher than in the other versions, with the CG having the lowest segregation. We further investigated the equilibria of the CG, in particular the influence of $\Delta$ and $\lambda$ on the obtained equilibria. See~\autoref{apx:experiments}.
Representative sample equilibria for all our variants of the Jump Game are depicted in~\autoref{fig:equijump}.
\begin{figure}[b]
    \centering
    \begin{subfigure}{0.31\linewidth}
        \includegraphics[width=\textwidth]{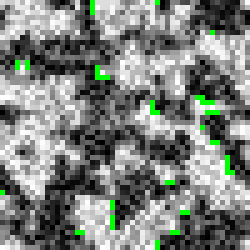}
        \caption{8-regular Maximum Distance Game}
        \label{fig:max8jump}
    \end{subfigure}
    \hfill
    \begin{subfigure}{0.31\linewidth}
        \includegraphics[width=\textwidth]{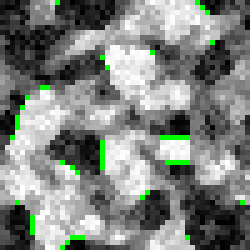}
        \caption{8-regular Average Distance Game} \label{fig:avg8jump}
    \end{subfigure}
    \hfill
    \begin{subfigure}{0.31\linewidth}
        \includegraphics[width=\textwidth]{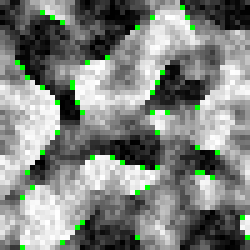}
        \caption{8-regular Cutoff Game}
        \label{fig:cutoff8jump}
    \end{subfigure}
    \caption{Equilibria in different Jump Games from an initially random
    starting strategy profile on 8-regular $50 \times 50$ toroidal grid
    graphs. 2\% of nodes are left empty and shown as green.
    }
    \label{fig:equijump}
\end{figure}
Here, a marked difference to the Swap Game becomes apparent: the equilibria seem to be less strongly segregated, than their respective counterpart in the Swap Game. This may be due to agents having fewer options to improve since at any time only a few empty cells are available. Moreover, note that in the equilibria the empty cells typically have neighboring agents with large type-difference, rendering these cells less attractive.

\paragraph{Quality Measures.}\label{sec:empquality}
We now focus on different quality measures to compare the obtained equilibria. For this, we use:
\begin{itemize}
\item \textbf{ADGSC}: the social cost of the ADG
\item \textbf{MDGSC}: the social cost of the MDG
\item $\#\leq\frac12$: the number of pairs of neighboring agents with type-difference at most $\frac12$;
\item \textbf{max $d$}: the maximum neighbor type-difference;
\item \textbf{Steps}: the number of steps until convergence.
\end{itemize}
The results, averaged
over $100$ runs from randomly chosen initial states, are shown in~\autoref{tab:quality}.
\begin{table}[t]
    \centering
    \resizebox{\linewidth}{!}{\begin{tabular}{lccccc}
        \toprule
        \textbf{Model} & \textbf{ADGSC} & \textbf{MDGSC} & \textbf{$\#\leq\frac12$} &
        \textbf{max $d$} & \textbf{Steps} \\
        \midrule
        S-MDG & 148 & 280 & 9979 & 0.75 & 16510\\
        S-ADG & 150 & 408 & 9779 & 0.97 & 9928\\
        S-CG, $0.1$ & 285 & 765 & 9340 & 0.98 & 2245\\
        S-CG, $0.2$ & 321 & 758 & 9446 & 0.99 & 1914\\
        \midrule
        J-UIS-MDG & 498 & 1010 & 9178 & 0.86 & 5059\\
        J-UIS-ADG & 281 & 647 & 9518 & 0.95 & 5498\\
        J-UIS-CG, $0.1$ & 227 & 592 & 9215 & 0.97 & 5451\\
        J-UIS-CG, $0.2$ & 284 & 638 & 9385 & 0.98 & 4470\\
        \bottomrule
    \end{tabular}}
    \caption{\label{tab:quality}
        Comparison of equilibria of our models. All
        values are averaged and rounded over 100 runs each.
    }
\end{table}
They confirm that the S-MDG produces the most segregated equilibria, followed by the S-ADG. Especially the S-MDG yields the lowest cost values in every cost function and on average has $99.79\%$ of edges between agents with type-distance at most~$\frac{1}{2}$. The MDG is also the only model where the maximum type difference between neighbors is significantly below 1. Also, our data indicate that equilibria in swap games are more segregated than the equilibria in jump games.
\paragraph{Discussion.} Our experiments shed light on the structural properties of the obtained equilibria. On the one hand, they reveal that the specific choice of cost
function can have a strong impact on the obtained states, e.g., that the CG yields very different outcomes compared to the MDG or the ADG, which is not obvious a priori. On the other hand, they also indicate that our model is rather robust with regard to certain kinds of cost functions, i.e., the similar structural results for the MDG and the ADG hint at the fact that varying the involved distance measures might not change the qualitative behavior much.
Also, the experiments
reveal that more structural properties might be analyzed in future work. E.g., the appearance of “hard borders”, i.e., the existence of many pairs of neighboring agents with
large type-value difference. Last but not least, such experiments also shed light on the preferences of real-world agents: while all our cost functions model homophily, our plots for the MDG and the ADG resemble segregation patterns that have been observed by sociologists, while the patterns of the
CG are different. This reveals what kind of homophilic behavior might be more realistic.

\section{Conclusion}
In this work, we study game-theoretic models for residential segregation with non-categorical agent types. This allows us to generalize existing models but also to derive novel results.

As a proof of concept, we have considered three very natural variants for the agents' behavior and focused on the most fundamental question: the existence of equilibria. For this, we present many positive results, in particular, we prove that SE in the S-MDG always exist and can be efficiently constructed on all graphs. We complete the picture by providing additional computational hardness results and many tight or almost tight bounds on the PoA and the PoS. Some interesting problems are left open, for example, settling the following:
\begin{conjecture}\label{simple_paths}
For any Swap Game played on a path, a SE which is also a social optimum can be efficiently computed.
\end{conjecture}

We emphasize that we have just explored a few of the many possible models using continuous type values. Future work could focus on different agent behavior, e.g., employing different norms to compare with. Also, as indicated by our simulations, exploring the relationship between the agents' behavior and the obtained segregation strength seems promising.

\bibliographystyle{named}
\bibliography{references}

\appendix

\section{Computational Complexity}\label{apx:complexity}
	In this section, we consider the complexity of computing a social optimum for
	different social cost functions. First, we show that finding social optima
	is NP-hard for all of our game variants. To achieve this, we use the reduction
	from \textsc{3-Partition} employed by \cite{E+19} to prove that finding
	an optimal strategy profile is NP-hard for a special case of the standard Swap Schelling
	Game with multiple types.

	The same reduction is applicable for all of our Swap Games. It uses the fact
	that agents receive the lowest cost if and only if their only neighbors are
	agents of equal types. In the S-ADG and S-MDG this
	is clearly the case. In the S-CG this can also be accomplished with an
	appropriate choice of $\lambda$. Thus, the hardness in our case immediately
	follows.

	\begin{corollary}
		Given a graph $G$ and a type function $t: [n] \to [0,1]$, finding an optimal
		strategy profile on 2-regular graphs in the S-MDG, S-ADG, and the S-CG is NP-hard.
	\end{corollary}

 Let the \emph{maximum edge cost} of $\sigma$ be defined as $\maxedge(\sigma) = \max_{i,j \in [n]: \{\sigma(i), \sigma(j)\} \in E} d(i,j)$. If there is no pair of adjacent agents under $\sigma$ $\maxedge(\sigma) = 1$. Note that for the MDG the maximum edge cost corresponds to the maximum cost of an agent.

	Note that the construction just mentioned has a strategy profile $\sigma$ with ${\maxedge(\sigma) = 0}$ if
	and only if the given instance is in \textsc{3-Partition}. Thus, we also get the
	result that finding a strategy profile with minimum maximum edge cost is NP-hard.

	\begin{corollary}
		Deciding if there is a strategy profile $\sigma$ with maximum edge cost $\maxedge(\sigma)
		= 0$ is NP-hard.
	\end{corollary}

	But even if we restrict to minimizing the maximum edge cost in a swap game,
	the problem remains NP-hard.

	\begin{theorem}
		Given a graph $G$, a type function $t: [n] \to [0,1]$ with $|V| =
		n$ agents, and a real number $s \in \mathbb{R}$, deciding if a strategy profile $\sigma$ with
		maximum edge cost $\maxedge(\sigma) \le s$ exists is NP-hard.
	\end{theorem}

	\begin{proof}
		We show this with a reduction from the \textsc{Bandwidth} problem shown to
		be NP-hard in \cite{bandwidth}, which is virtually the same problem. In the
		\textsc{Bandwidth} problem, an instance consists of a graph $G = (V,E)$ and
		a number $k \in \mathbb{N}$. The problem asks whether there is a bijection $\ell: V
		\to [ |V| ]$, such that $\max_{\{u,v\} \in E}|\ell(u) - \ell(v)| \le k$. This
		bijection directly translates to a strategy profile in our problem by assigning a
		node $v\in V$ to the agent with type $\frac{\ell(v) - 1}{|V| - 1}$ and the
		another way around. Thus, a \textsc{Bandwidth}-instance directly corresponds
		to an instance in our problem by setting $s = \frac{k}{|V| - 1}$, which
		completes the proof.
	\end{proof}

\section{Omitted Details from Section~\ref{sec:existence}}\label{apx:existence}

\propthree*
\begin{proof}
			We show this proposition using a more general version of the ordinal potential function
			introduced by \cite{CLM18}. Let $G$ be a $\Delta$-regular graph. We define our ordinal potential function as the sum of distances over all edges in the graph according to a given strategy profile
			$\sigma$: \[\Phi(\sigma) =
			\sum_{\{u, v\} \in E} d(\sigma^{-1}(u), \sigma^{-1}(v)) = 2 \cdot
			\scost(\sigma).\] We now
			show that this is an ordinal potential function. Let $i, j \in [n]$ be two
			agents performing an improving swap in $\sigma$. Therefore, both costs decrease and we have
			\begin{equation} \label{eq:costa}
				\sum_{k \in N_i(\sigma_{ij})} d(i,k) <
				\sum_{k \in N_i(\sigma)} d(i,k)
			\end{equation}
			and
			\begin{equation} \label{eq:costb}
				\sum_{k \in N_j(\sigma_{ij})} d(j,k) <
				\sum_{k \in N_j(\sigma)} d(j,k).
			\end{equation}
			Notice that this is just the cost function of $i$ and $j$ multiplied by $\Delta$. Because
			agents only influence their neighborhood and $\sigma_{ij}(i) = \sigma(j)$ and $\sigma_{ij}(j)
			= \sigma(i)$, we get
			\begin{align*}
				\Phi(\sigma_{ij}) = \Phi(\sigma) & - \left(\sum_{k \in N_i(\sigma)}
				d(i,k) + \sum_{k \in N_j(\sigma)} d(j,k)\right)           \\
				& + \sum_{k \in N_i(\sigma_{ij})}
				d(i,k) + \sum_{k \in N_j(\sigma_{ij})} d(j,k).
			\end{align*}

			\noindent By~\autoref{eq:costa} and~\autoref{eq:costb}, this change is negative and we get
			our desired result $\Phi(\sigma_{ij}) < \Phi(\sigma)$ which makes the S-ADG a potential game.
		\end{proof}

\propfour*
\begin{proof}
Let $G$ be an almost $\Delta$-regular graph.
Consider the function $\Phi(\sigma) = \lvert \{u,v\}\in E \colon \sigma^{-1}(u)\in N_{\sigma^{-1}(v)}^+(\sigma)\}\lvert$ which counts the number of \textit{monochromatic} edges, i.e., all edges whose endpoints are occupied by agents who are friends.
We prove the claim by showing that $\Phi$ is an ordinal potential function.
Let $i,j\in [n]$ be two agents performing an improving swap.

Observe that a swap between $i$ and $j$ only influences their direct neighborhood.
Therefore $\Phi(\sigma_{ij})-\Phi(\sigma) = \lvert N_i^+(\sigma_{ij})\rvert + \lvert N_j^+(\sigma_{ij})\rvert - \lvert N_i^+(\sigma)\rvert - \lvert N_i^+(\sigma)\rvert$ and we want to prove $\Phi(\sigma_{ij})-\Phi(\sigma)>0$ by showing that $\lvert N_i^+(\sigma_{ij})\rvert + \lvert N_j^+(\sigma_{ij})\rvert - \lvert N_i^+(\sigma)\rvert - \lvert N_i^+(\sigma)\rvert > 0$.

First, consider the case $\deg(\sigma(i))=\deg(\sigma(j))$.
Because the swap is improving for both agents, we have
\begin{align}
    \frac{\lvert N_i^+(\sigma)\rvert}{\Delta} < \frac{\lvert N_i^+(\sigma_{ij})\rvert}{\Delta},
\end{align} and
\begin{align}
    \frac{\lvert N_j^+(\sigma)\rvert}{\Delta} < \frac{\lvert N_j^+(\sigma_{ij})\rvert}{\Delta}.
\end{align}
This implies \[\lvert N_i^+(\sigma_{ij})\rvert + \lvert N_j^+(\sigma_{ij})\rvert - \lvert N_i^+(\sigma)\rvert - \lvert N_i^+(\sigma)\rvert >0\]
as desired.

Second, consider the case $\deg(\sigma(i))\neq \deg(\sigma(j))$. Assume, without loss of generality, $\deg(\sigma(i))=\Delta$ and $\deg(\sigma(j))=\Delta +1$.
Because the swap is improving for both agents, we have
\begin{align}\label{i-equation}
    \frac{\lvert N_i^+(\sigma)\rvert}{\Delta} < \frac{\lvert N_i^+(\sigma_{ij})\rvert}{\Delta + 1},
\end{align} and
\begin{align}\label{j-equation}
    \frac{\lvert N_j^+(\sigma)\rvert}{\Delta + 1} < \frac{\lvert N_j^+(\sigma_{ij})\rvert}{\Delta}.
\end{align}
From~\autoref{i-equation} we derive $\lvert N_i^+(\sigma)\rvert < \lvert N_i^+(\sigma_{ij})\rvert$.
We show that~\autoref{j-equation} implies $\lvert N_j^+(\sigma)\rvert \leq \lvert N_j^+(\sigma_{ij})\rvert$.
Let $x=\lvert N_j^+(\sigma)\rvert$.
Assume towards contradiction $\lvert N_j^+(\sigma)\rvert > \lvert N_j^+(\sigma_{ij})\rvert$, i.e., $x-y = \lvert N_j^+(\sigma_{ij})\rvert$ for some $y \in \mathbb{N}^+$. If we plug this into~\autoref{j-equation}, we get
\begin{alignat*}{2}
    &\frac{x}{\Delta + 1} &&< \frac{x-y}{\Delta} \\
    \Longleftrightarrow\quad &\Delta x &&< (x-y) (\Delta +1) \\
    \Longleftrightarrow\quad &\Delta x &&< \Delta x + x - y (\Delta +1) \\
    \Longleftrightarrow\quad & y (\Delta +1)&&<x,
\end{alignat*}
but $x=\lvert N_j^+(\sigma)\rvert\leq \deg(\sigma(j)) = \Delta + 1$. This is a contradiction.
Therefore we have
\[\lvert N_i^+(\sigma_{ij})\rvert + \lvert N_j^+(\sigma_{ij})\rvert -\lvert N_i^+(\sigma)\rvert - \lvert N_j^+(\sigma)\rvert  >0. \qedhere\]
\end{proof}

\theoremjumppathext*
\begin{proof}
First of all, observe that, by construction, the JE $\sigma$ for J-HIS-MDGs returned by the algorithm defined in the proof of~\autoref{jump-path} verifies the following property. Let $i$ be a non-isolated agent adjacent to an empty spot in $\sigma$ and let $j$ be her unique neighbor in $\sigma$. For any empty spot separating two agents $k$ and $k+1$, it holds that $d(i,j)\leq d(k,k+1)$.

Consider a J-HIS-ADG and fix an agent $i$. Since $\sigma$ is a JE for the related J-HIS-MDG, the maximum distance of $i$ towards her neighbors does not decrease after a jump. So, in order for the average distance to decrease after jumping, the minimum distance of $i$ towards her neighbors needs to decrease after jumping. Assume that $i$ jumps to an empty spot located at the right of her current position along the path. For the minimum distance of $i$ to decrease, as $\sigma$ locates the agents along the path in an ordered-based way, $i$ has to be at the left of the empty spot in $\sigma$ and $i$ must have two neighbors in $\sigma_i$ having consecutive indices, say $j$ and $j+1$. So, since $i-1$ is the unique neighbor of $i$ in $\sigma$, we have $d(i,i-1)>d(i,j)\geq d(i,i+1)$, which contradicts the property stated above. If~$i$ jumps to an empty spot left of her current position along the path, a symmetric argument applies.

Consider a J-HIS-CG and fix an agent $i$. In a CG on a path, an agent can only pay $0$, $1/2$, or $1$. In order for $i$ to perform an improving swap, she must be paying either $1/2$ or $1$ in $\sigma$. If $i$ is paying $1/2$, she must end up paying~$0$ in~$\sigma_i$, meaning that her maximum distance towards her neighbors has to decrease after jumping, thus contradicting the fact that $\sigma$ is a JE for the related J-HIS-MDG. If $i$ is paying $1$ and ends up paying $0$ in $\sigma_i$, again her maximum distance has to decrease after jumping which generates a contradiction. So, $i$ is paying $1$ and ends up paying $1/2$ in $\sigma_i$. This means that the minimum distance of $i$ towards her neighbors needs to decrease after jumping falling within the same case considered above for J-HIS-ADGs.
\end{proof}

\section{Quality of Equilibria}\label{apx:quality}
In this section, we provide an overview of our results for the $\hpoa$ and $\hpos$ of the considered games. For the former, in particular, we were able to provide a full characterization, while, for the latter, we give results for games played on specific topologies.  Our findings are summarized in~\autoref{tab:poa} and~\autoref{tab:pos}.

\subsection{Price of Anarchy}
\begin{table}[b]
    \centering
\resizebox{0.48\textwidth}{!}{\begin{tabular}{@{}lrrr}
       \toprule
       $\hpoa$ & MDG & ADG & CG \\
       \midrule
       Swap
       & $\Theta(n)$ (\cref{thm:Swap_MDG_PoA}) & $\Theta(\Delta n)$ (\cref{thm:Swap_ADG_PoA}) & $\infty$ (\cref{thm:Swap_CG_PoA}) \\
       Jump UIS
       & $\infty$ (\cref{thm:Jump_MDG_PoS}) & $\infty$ (\cref{thm:Jump_ADG_PoA}) & $\infty$ (\cref{thm:Swap_CG_PoA}) \\
       Jump HIS
       & $\infty$ (\cref{thm:Jump_HISMDG_PoA}) & $\infty$ (\cref{thm:Jump_ADG_PoA}) & $\infty$ (\cref{thm:Swap_CG_PoA}) \\
       \bottomrule
   \end{tabular}}
    \caption{Bounds on the Price of Anarchy. The maximum degree of the underlying graph topology is denoted by $\Delta$.}
   \label{tab:poa}
\end{table}
Under a worst-case perspective, we can generally prove the tremendously bad performance of equilibria, as there can be equilibria of positive social cost even when there exists a social optimum attaining a social cost equal to zero, which results in an unbounded $\hpoa$. The only exceptions are S-MDGs and S-ADGs for which the $\hpoa$ becomes $\Theta(n)$ and $\Theta(n\Delta)$, respectively, where $\Delta$ denotes the maximum degree of the underlying graph topology.

\subsubsection{Maximum Type-Distance Game}
    \begin{theorem} \label{thm:Swap_MDG_PoA}
    The $\hpoa$ for the S-MDG is in $\Theta(n)$. The lower bound holds also for games played on paths.
    \end{theorem}
    \begin{proof}

   For any strategy profile $\sigma$, let $i$ be the agent with the largest index among the ones which are adjacent to agent $1$ in $\sigma$, which yields $\cost 1(\sigma)=t(i)-t(1)$. Now, let $i'$ be the agent with the largest index among the ones which are adjacent to the subgraph induced by $\sigma(1)$ and all of her neighbors. There is an agent residing in this subgraph other than agent $1$ that has to pay at least $t(i')-t(i)$. By iterating this reasoning, we get $\scost(\sigma)\geq t(n)-t(1)$ for any strategy profile $\sigma$. Since it also holds that $\scost(\sigma)\leq n(t(n)-t(1))$ for any strategy profile $\sigma$, it follows that the $\hpoa$ is at most $n$.

    For the lower bound, let $G$ be an $n$-node path, with $n$ even, and consider $n$ agents half of which have type $0$, and the remaining ones have type $1$. The social optimum is attained by listing all type-$0$ agents followed by all type-$1$ agents for a social cost of $2$. A SE $\sigma$ of social cost $n-2$ can be obtained by replicating the sequence made of two type-$0$ agents followed by two type-$1$ agents. In $\sigma$, all agents except for the first and the last are paying a non-zero cost which amounts to $1$. To see that this is a SE, observe that the only way for an agent to improve is to end up paying a cost of zero which requires being surrounded by agents of the same type. But this may only happen when swapping with either the first or the last agent of the sequence who, being paying zero, are not interested in swapping.
    \end{proof}

    \begin{theorem} \label{thm:Jump_HISMDG_PoA}
    The $\hpoa$ for the J-HIS-MDG is unbounded.
    \end{theorem}
    \begin{proof}
    Let $G_1$ be a clique of four nodes with a distinguished node denoted by $u$, and let $G$ be the graph obtained from $G_1$ by adding two nodes $v$ and $w$ such that $v$ is adjacent to $u$ and $w$ is adjacent to $v$. Consider the J-HIS-MDG defined over $G$ in which there are three agents of type $0$ and two agents of type $
    1$. The strategy profile $\sigma^*$, in which the three agents of type $0$ are located on the three nodes of $G_1$ other than $u$, has $\scost(\sigma^*)=0$. On the other hand, the strategy profile $\sigma$ in which the three agents of type $0$ are located on $u$, $v$, and $w$ is an equilibrium such that $\scost(\sigma)=3$. So, the claim follows.
    \end{proof}

     \subsubsection{Average Type-Distance Game}

\begin{theorem} \label{thm:Swap_ADG_PoA}
    			The $\hpoa$ for the S-ADG is in
    			$\Theta(\Delta n)$, where $\Delta$ is the maximum degree of the underlying graph.
    		\end{theorem}
      \begin{proof}
      For the upper bound, we can use the same argument exploited in the proof of \autoref{thm:Swap_MDG_PoA}. For any strategy profile $\sigma$, let $i$ be the agent with the largest index among the ones which are adjacent to agent $1$ in $\sigma$, which yields $\cost 1(\sigma)\leq\frac{t(i)-t(1)}{\Delta}$. Now, let $i'$ be the agent with the largest index among the ones which are adjacent to the subgraph induced by $\sigma(1)$ and all of her neighbors. There is an agent residing in this subgraph other than agent $1$ that has to pay at least $\frac{t(i')-t(i)}{\Delta}$. By iterating this reasoning, we get $\scost(\sigma)\geq \frac{t(n)-t(1)}{\Delta}$ for any strategy profile $\sigma$. Since it also holds that $\scost(\sigma)\leq n(t(n)-t(1))$ for any strategy profile $\sigma$, it follows that also the $\hpoa$ under the social function $\scost$ is at most $\Delta n$.

      \begin{figure}
          \centering
          \includegraphics[scale=0.15]{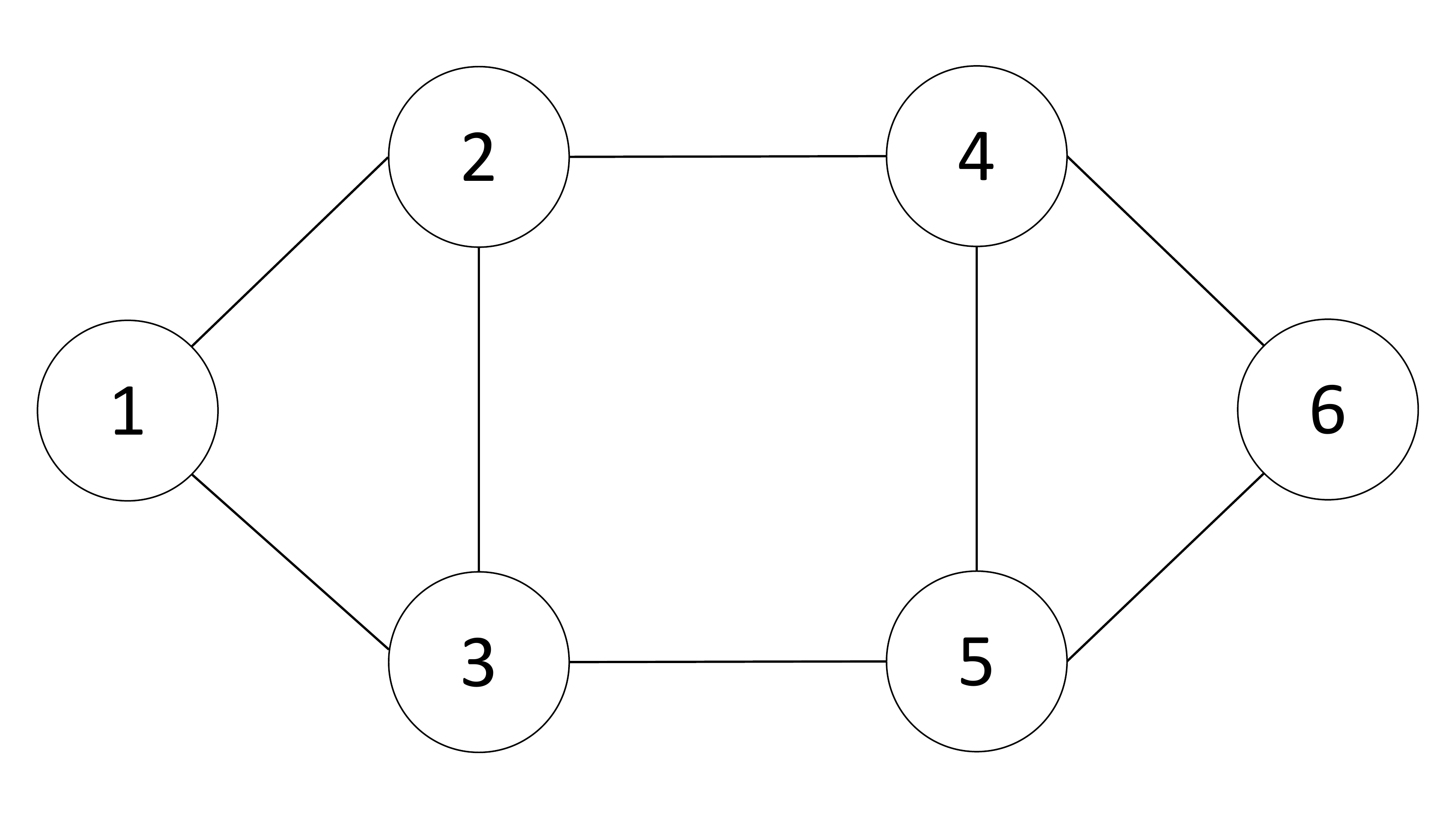}
          \caption{Graph $G_i$, $i\in [2h]$, used as a gadget in the proof of \autoref{thm:Swap_ADG_PoA}. Node $1$ is called the {\em leftmost node}, node $6$ is called the {\em rightmost node}, and all other nodes are called {\em middle nodes}.}
          \label{fig-graph}
      \end{figure}

      For the lower bound, consider a graph $G$ obtained by the union of $2h+1$ subgraphs $G_0,G_1,\ldots,G_{2h}$. $G_0$ is a clique of $k$ nodes with two distinguished nodes $x$ and $y$ and, for every $i\in [2h]$, $G_i$ is the graph depicted in~\autoref{fig-graph}. $G$ is obtained as follows: graphs $G_1,\ldots,G_h$ are joined by sharing their leftmost node which gets called $u$; similarly, graphs $G_{h+1},\ldots,G_{2h}$ are also joined by sharing their left node which gets called $v$; finally, $G$ is made connected by creating edges $\{u,x\}$ and $\{v,y\}$. So, $G$ has a total of $n=k+10h+2$ nodes and $\Delta=k$. We choose $k$ and $h$ such that $k$ is even and $k=10h+2$, so that $k=n/2\geq 4$ and $h=n/20-2$. The $n$ agents are such that half of them have type $0$ and the remaining half have type $1$.

      Consider the strategy profile $\sigma$ in which $G_0$ contains $k/2$ agents of type $0$ and $k/2$ agents of type $1$ in such a way that $x$ is assigned an agent of type $0$ and $y$ an agent of type $1$. Each graph $G_i$ with $i\in [h]$, has an agent of type $0$ on the common leftmost node, $2$ agents of type $0$ on nodes $2$ and $3$, and $3$ agents of type $1$ on the remaining $3$ ones. Each graph $G_i$ with $i\in [2h]\setminus [h]$, has an agent of type $1$ on the common leftmost node, $2$ agents of type $1$ on nodes $2$ and $3$, and $3$ agents of type $0$ on the remaining $3$ ones.

      We shall prove that $\sigma$ is a SE. Consider a swap between two agents $i$ and $j$ of types $0$ and $1$, respectively. These agents cannot both belong to $G_0$. In fact, if $\sigma(i)\neq x$ and $\sigma(j)\neq y$, $i$ and $j$ do not change their neighborhood by swapping since $G_0$ deprived of both $x$ and $y$ is a clique with no edges towards the other subgraphs. If $\sigma(i)=x$, then $\cost j(\sigma_{ij})=\frac{k+2}{2k}$. If $\sigma(j)=y$, then $\cost j(\sigma)=1/2$, while, if $\sigma(j)\neq y$, then $\cost j(\sigma)=\frac{k}{2k-2}$. In both cases, $\cost j(\sigma)\leq\cost j(\sigma_{ij})$, so $j$ is not interested in swapping. Outside $G_0$, only the four middle nodes of each graph $G_\ell$, $\ell\in [2h]$, are paying a non-zero cost and, more precisely, they are all paying a cost of $1/3$. So, they are the only ones who can be interested in swapping. These agents all have $3$ neighbors, so, if a swap occurs between them, the only way to improve on a cost of $1/3$ is to get a cost of $0$ in $\sigma_{ij}$, which is not possible. The only left possible case is a swap between an agent located in a middle node of a graph $G_\ell$, $\ell\in [2h]$, and an agent in $G_0$. However, whoever ends up in $G_0$ pays at least $\frac{k-2}{2(k-2)}$ which is never better than $1/3$. So, $\sigma$ is a SE. The social cost of $\sigma$ is at least $\frac{k}{2}+\frac{8h}{3}=\Omega(\Delta)$.

      In the strategy profile in which all agents of type $0$ are assigned to $G_0$, only the agents located at $x$, $y$, $u$ and $v$ pays a non-zero cost for a social cost of $\frac{2}{k}+\frac{2}{h+1}=O(1/n)$. Thus a lower bound of $\Omega(\Delta n)$ of the $\hpoa$ follows.
      \end{proof}

    \begin{theorem} \label{thm:Jump_ADG_PoA}
    			The $\hpoa$ for both the J-UIS-ADG and the J-HIS-ADG is unbounded even on paths.
    		\end{theorem}
      \begin{proof}
      Let $G$ be a $7$-node path and consider $6$ agents such that $t(1)=\ldots=t(4)$ and $t(5)=t(6)=1$; thus, $e=1$. Let $\sigma^*$ be the strategy profile realizing the sequence $(1,2,3,4,\text{\o},5,6)$. All nodes are paying a cost of zero, no matter which variant is used, as no agent is isolated.

      Now let $\sigma$ be the strategy profile realizing the sequence $(3,4,5,6,\text{\o},1,2)$. Also in this case, as no agent is isolated, the cost of each agent is not influenced by the chosen variant. The only agents not paying zero are agents $4$ and $5$, both paying $1/2$. By jumping to the empty spot, both agents end up paying $1/2$, so $\sigma$ is a JE. As its social cost is $1$ the claim follows.
      \end{proof}

\subsubsection{Cutoff Games}
\begin{theorem} \label{thm:Swap_CG_PoA}
    The $\hpoa$ for all CGs is unbounded
    even for games played on paths.
    \end{theorem}
    \begin{proof}
    Let us start with the S-CG. Let $G$ be a $6$-node path and define $t(i)=\frac{i-1}{5}$ for each $i\in [6]$. Let the cutoff parameter be $\co=2/5$. The strategy profile placing all agents along the path in an ordered-base way has social cost of $0$ for both social functions. We claim that the strategy profile $\sigma$ creating the sequence $(1,3,4,5,6,2)$ is a SE. Indeed, the only agents not paying $0$ in $\sigma$ are agents $2$ and $6$ paying, respectively, a cost of $1$ and $1/2$. So they are the only two candidates to perform an improving swap in $\sigma$. However, after swapping, the agent $6$ ends up paying $1$. So, $\sigma$ is a SE. As $\scost(\sigma)=3/2$, the claim follows.

    To extend the instance to the J-UIS-CG, extend the path with another node and keep the same agents so that $e=1$. The social optimum does not change if one puts the empty spot after agent $6$. Now extend strategy profile $\sigma$ by adding an empty spot after agent $2$. The cost of all agents does not change, so the only candidates to perform an improving jump are still agents $2$ and $6$, but none of the improves by jumping.

    To cover the J-HIS-CG, modify the instance for the J-UIS-CG by redefining $t(4)$ as $\frac{3}{5}-\epsilon$, with $\epsilon>0$ arbitrarily small. This minor modification does not influence the social optimum. Now let $\sigma$ by the strategy profile realizing the sequence $(1,3,5,4,6,\text{\o},2)$. All agents pay $0$, except for agents $4$ and $6$ paying, respectively, $1/2$ and $1$. By jumping, $2$ ends up paying $1/2$ and $6$ ends up paying $1$, so no improvements are possible.
    \end{proof}

\subsection{Price of Stability}
For the characterization of the $\hpos$, i.e., best-case performance of equilibria, it is usually required either the FIP or the existence of (not necessarily efficient) algorithms computing equilibria with provable approximation guarantees. As we have seen in~\autoref{sec:existence}, this may be either impossible or require quite an effort; nevertheless, these difficulties are common also in previous models of the Schelling game.
\begin{table*}[t]
    \centering
    \resizebox{\textwidth}{!}{\begin{tabular}{@{}lrrrrrrrrr}
        \toprule
        \multirow{2}{4em}{$\hpos$}& \multicolumn{3}{c}{MDG} & \multicolumn{3}{c}{ADG} & \multicolumn{3}{c}{CG} \\
        \cmidrule(l){2-4} \cmidrule(l){5-7}\cmidrule(l){8-10}
        & path & regular & general & path & regular & general & path & regular & general \\
        \midrule
        Swap
        &  $\leq2$ (\cref{thm:Swap_MDG_PoS_path})&   & $\Theta(n)$ (\cref{thm:Swap_MDG_PoS_regular})
        & 1 (\cref{PoS-AVG-path}) & 1 (\cref{thm:Swap_ADG_PoS}) &
        & $\leq2$ (\cref{thm:Swap_CP_PoS_path}) & 1 (\cref{thm:Swap_CG_PoS}) &  \\
        Jump UIS
        &  &  & $\infty$ (\cref{thm:Jump_MDG_PoS})
        &  &  &
        &  &  &  \\
        Jump HIS
        &  $\leq2$ (\cref{thm:JumpHIS_MDG_PoS_path}) &  & $\Theta(n)$ (\cref{thm:jumpHIS_MDG_PoS_regular})
        &   &  &
        &   &  &  \\
        \bottomrule\\
    \end{tabular}}
    \caption{Bounds on the Price of Stability.}
    \label{tab:pos}
\end{table*}

\subsubsection{Maximum Type-Distance Game}

    \begin{theorem}\label{thm:Swap_MDG_PoS_regular} \label{thm:jumpHIS_MDG_PoS_regular}
    The $\hpos$ for both the S-MDG and the J-HIS-MDG is in $\Theta(n)$.
    \end{theorem}
    \begin{proof}
    From~\autoref{thm:Swap_MDG_PoA} and the fact that the $\hpos$ cannot be worse than the $\hpoa$, we have that the $\hpos$ is $O(n)$ for the S-MDG. For the J-HIS-MDG, by~\autoref{thm:Jump_HISMDG_Ex}, we know that there exists a sequence of improving jumps starting from $\sigma^*$ and ending at an equilibrium $\sigma$ such that $\max_{i\in [n]}\cost i(\sigma)\leq\max_{i\in [n]}\cost i(\sigma^*)$. So, we get that the $\hpos$ is $O(n)$ by exploiting the relations $\scost(\sigma)\leq n \max_{i\in [n]}\cost i(\sigma)\leq n \max_{i\in [n]}\cost i(\sigma^*)\leq n \scost(\sigma^*)$.

    To show the asymptotically matching lower bounds, we start with the S-MDG and then discuss how to extend the lower bound also to the J-HIS-MDG.

    Consider four cliques $K_1,\ldots, K_4$ such that $|K_i|=m+i$ for each $i\in [4]$ and let us denote with $u_i$ an arbitrary fixed node in $K_i$. Let $G$ be the connected graph with $n=\sum_{i\in [4]}|K_i|=4m+10$ nodes obtained by creating the path $\langle u_1,u_2,u_3,u_4\rangle$. Define a S-MDG by creating $|K_i|$ agents of type $t(i)$, for each $i\in [4]$, with $t(1)=1/2$, $t(2)=1$, $t(3)=\varepsilon$ and $t(4)=2/3$, where $\varepsilon>0$ is arbitrarily small. Observe that, for each $i\in [4]$, clique $K_i$ can precisely allocate all agents of type $i$. The strategy profile $\sigma^*$ satisfying this property has $\scost(\sigma^*)=\frac{19}{6}-3\varepsilon=\Theta(1)$. This profile is not an equilibrium as the agent of type $1$ assigned to node $u_1$ can profitably swap with the agent of type $4$ assigned to node $u_4$. In fact, in $\sigma^*$, the former agent is paying $1/2$, while the latter is paying $2/3-\varepsilon$. By swapping, they end up paying $1/2-\varepsilon$ and $1/3$, respectively. The $\Omega(n)$ lower bound follows by observing that an equilibrium always exists and that any strategy profile other than $\sigma^*$ needs to mix up agents of different types in at least one clique. When this happens, as $|K_i|\geq m+1$ for every $i\in [4]$ and $d(i,j)\geq 1/6$ for every two distinct $i,j\in [4]$, the social cost has to be at least $(m+1)/6=\Omega(n)$, for every non-constant value of $m$.

    To extend this lower bound to the J-HIS-MDG, it suffices to define $|K_4|=m+5$, i.e., adding a node to $K_4$. Again, the only strategy profile (up to a permutation of the nodes in $K_4$) guaranteeing a social cost $o(n)$ is $\sigma^*$ as defined above with an empty node left in $K_4$. This profile is not an equilibrium as the agent in $u_1$ is paying $1/2$ and, by jumping to the empty node in $K_4$ ends up paying at most $1/2-\varepsilon$.
    \end{proof}

\begin{theorem} \label{thm:Swap_MDG_PoS_path}
The $\hpos$ for the S-MDG on paths is at most 2.
\end{theorem}
\begin{proof}
Let $\sigma$ be a strategy profile in which agent $i$ is placed on the $i$-th node of the path, where the nodes of the path are numbered from left to right.

One can easily observe that $\sigma$ is a SE. Let $\sigma^*$ be a strategy profile of minimum total cost. We show that $\scost(\sigma) \leq 2\cdot\scost(\sigma^*)$. We can upper bound the cost of $\sigma$ by
\begin{eqnarray*}
\scost(\sigma) & = & (t(2)-t(1)) + (t(n)-t(n-1))\\
& & \, + \sum_{i = 2}^{n-1} \max\big\{t(i)-t(i-1), t(i+1)-t(i)\big\}\\
& \leq & (t(2)-t(1)) + (t(n)-t(n-1))\\
& & \,  + \sum_{i = 2}^{n-1} \big((t(i)-t(i-1))+(t(i+1)-t(i))\big)\\
& = & 2(t(n)-t(1)).
\end{eqnarray*}
Let $t_i$ be the type of the agent that occupies node $i$ in $\sigma^*$, where nodes are numbered in left-to-right order from 1 to $n$. We lower bound the cost of $\sigma^*$ by
\begin{eqnarray*}
\scost(\sigma^*) & = & |t_2-t_1| + |t_{n}-t_{n-1}|\\
& & \ \ \ + \sum_{j=2}^{n-1} \max\{|t_j-t_{j-1}|,|t_{j+1}-t_j|\} \\
& \geq & |t_2-t_1| + |t_{n}-t_{n-1}|\\
& & \ \ \ + \frac{1}{2}\sum_{j=2}^{n-1} \big(|t_j-t_{j-1}|+|t_{j+1}-t_j|\big) \\
& \geq & \sum_{j=2}^{n}|t_j-t_{j-1}| \geq  t(n)-t(1).
\end{eqnarray*}
The claim follows by combining the upper bound on $\scost(\sigma)$ with the lower bound on $\scost(\sigma^*)$.
\end{proof}

    \begin{theorem} \label{thm:Jump_MDG_PoS}
    The $\hpos$ (and so also the $\hpoa$) for the J-UIS-MDG is unbounded.
    \end{theorem}
    \begin{proof}
    For an integer $m\geq 3$, let $G$ be a graph with $3m+1$ nodes organized as follows: there are $3$ disjoint cliques made of $m$ nodes each, that we denote by $G_1$, $G_2$, and $G_3$; then, there is a node $u$ which is adjacent to all other nodes of $G$. So, $G$ is connected.

    Now, consider the J-UIS-MDG defined by $G$ in which there are $m$ agents of type $0$, $m$ agents of type $1$, and $1$ agent of type $1-\varepsilon$, for an arbitrarily small $\varepsilon>0$.

    Let us first observe that there exists a strategy profile $\sigma^*$ such that $\scost(\sigma^*)$ is in $o(1)$, which is obtained by placing all agents of type $0$ in $G_1$, $m-1$ agents of type $1$ in $G_2$ and the two remaining agents in $G_3$. We shall prove the claim by showing that the game admits at least an equilibrium and that, for any equilibrium $\sigma$, $\scost(\sigma)$ is in $\Omega(1)$.

    Towards this end, observe that to have a social cost $o(1)$, no agent can be assigned to node $u$. Moreover, the agent of type $1-\varepsilon$ needs to be adjacent to agents of type $1$ only. So, we can assume, without loss of generality, that the agent of type $1-\varepsilon$ is assigned to a node of $G_1$ together with $x$ agents of type $1$, where $1\leq x <m$. As no agent can be assigned to $u$, there must be at least one agent of type $1$ assigned to a node of $G_i$, with $i\in\{2,3\}$. To guarantee a social cost $o(1)$, no agent of type $0$ can be assigned to a node of $G_i$. So, there are at least $x\geq 1$ empty nodes in $G_i$. Any of the $x\geq 1$ agents of type $1$ in $G_1$ is paying a cost of $\varepsilon$ and, by jumping to an empty node in $G_i$, it ends up paying $0$. So, no equilibria can guarantee a social cost of $o(1)$.

    Finally, it is easy to check that the strategy profile in which all agents of type $0$ are placed on $G_1$, all agents of type $1$ are placed on $G_2$ and the agent of type $1-\varepsilon$ is placed on $u$ is an equilibrium.
    \end{proof}

In the following, we analyze the PoS for J-HIS-MDG. Let us call a \emph{block} any maximal subpath whose nodes are all occupied by agents according to a given strategy profile. We first show how to lower bound the social value of a strategy profile.

\begin{lemma}\label{lemma:lower_bound_sum}
Let $\sigma$ be a strategy profile of a J-HIS-MDG played on a path. Let $B_1,\ldots, B_h$ be the blocks induced by $\sigma$. For each block $B_i$, let $m_i$ (resp., $M_i$) be the minimum (resp., maximum) type of an agent that occupies a node of $B_i$. We have that $\scost(\sigma) \geq \sum_{i=1}^h (M_i-m_i)$.
\end{lemma}
\begin{proof}
We prove the claim by showing that the overall sum of the costs of the agents in block $B_i$, that we denote by $\scost_i(\sigma)$, is at least $M_i - m_i$, for every $i\in [h]$. Assume that block $B_i$ contains $n'$ agents and the $j$-th node of $B_i$ has type $t_j$. The value $\scost_i(\sigma)$ can be lower bounded by the following
\begin{eqnarray*}
\scost_i(\sigma) & = & |t_2-t_1| + |t_{n'}-t_{n'-1}|\\
& & \ \ \ + \sum_{j=2}^{n'-1} \max\{|t_j-t_{j-1}|,|t_{j+1}-t_j|\} \\
& \geq & |t_2-t_1| + |t_{n'}-t_{n'-1}|\\
& & \ \ \ + \frac{1}{2}\sum_{j=2}^{n'-1} \big(|t_j-t_{j-1}|+|t_{j+1}-t_j|\big) \\
& \geq & \sum_{j=2}^{n'}|t_j-t_{j-1}| \geq  M_i-m_i,
\end{eqnarray*}
where the last inequality holds because, for any $t \in [m_i,M_i]$, there exists one interval $[a,b]$, with $a=\min\{t_{j-1}, t_j\}$ and $b=\max\{t_{j-1},t_j\}$, such that $t \in [a,b]$.\footnote{
To see this, consider the $n'$ types $t_1,\ldots,t_{n'}$ as points on a line and add all the edges $(t_j,t_{j+1})$, for every $j\in [n'-1]$. The obtained graph is connected and, clearly, each edge models an interval $[a,b]$, with $a=\min\{t_{j}, t_{j+1}\}$ and $b=\max\{t_{j},t_{j+1}\}$, for some $j\in [n'-1]$. Any $t \in (m_i,M_i)$ defines a cut $U=\{t_j \leq t \mid j\in [n']\}$ and $\bar U=\{t_j > t \mid j\in [n']\}$ which is clearly crossed by an edge.}
\end{proof}
Next, we upper bound the cost of any strategy profile which sorts the agent in non-decreasing order within each block.
\begin{lemma}\label{lemma:upper_bound_sum}
Let $\sigma$ be a strategy profile of a J-HIS-MDG played on a path. Let $B_1,\ldots, B_h$ be the blocks induced by $\sigma$. For each block $B_i$, let $m_i$ (resp., $M_i$) be the minimum (resp., maximum) type of an agent that occupies a node of $B_i$. Moreover, for every $i$, assume that agents which occupy the nodes of $B_i$ are placed in non-decreasing order in $\sigma$. We have that $\scost(\sigma) \leq 2\sum_{i=1}^h (M_i-m_i)$.
\end{lemma}
\begin{proof}
We prove the claim by showing that the overall sum of the costs of the agents in block $B_i$, again denoted by $\scost_i(\sigma)$, is at most $2(M_i - m_i)$, for any $i\in [h]$. Let $n'$ be the number of agents that occupy the nodes on $B_i$. Denote by $t_j$ the type of the agent that occupies the $j$-th node of $B_i$. Clearly, $t_1=m_i$ and $t_{n'}=M_i$. We bound the value $\psi$ as follows
\begin{eqnarray*}
\scost_i(\sigma) &=& (t_2-t_1) + (t_{n}-t_{n-1})\\
& & \ \ \ + \sum_{j=2}^{n-1} \max\{t_j-t_{j-1},t_{j+1}-t_j\} \\
& \leq & (t_2-t_1) + (t_{n'}-t_{n'-1})\\
& & \ \ \ + \sum_{j=2}^{n'-1} \big((t_j-t_{j-1})+(t_{j+1}-t_j)\big) \\
& = & 2\sum_{j=2}^{n'}(t_j-t_{j-1}) = 2(M_i-m_i). \hspace*{0.5cm}\qedhere
\end{eqnarray*}
\end{proof}

By leveraging~\autoref{lemma:lower_bound_sum} and~\autoref{lemma:upper_bound_sum}, we can prove the following result.

\begin{theorem}\label{thm:JumpHIS_MDG_PoS_path}
The $\hpos$ for the J-HIS-MDG on paths is at most $2$.
\end{theorem}
\begin{proof}
Clearly, if the number $e$ of empty nodes satisfies $e\geq n-1$, a strategy profile of social cost equal to  zero can be obtained by making every agent isolated. Such a profile is trivially a JE and a social optimum. So, assume that $e<n-1$.

Let $S$ be a set of $e$ pairs of consecutive agents yielding the largest intervals occurring between the types of two consecutive agents, and let $\sigma$ be the strategy profile obtained by placing the agents in increasing order along the path and leaving an empty spot between any two agents in $S$. We already proved in~\autoref{jump-path} that $\sigma$ is a JE.

Let $\sigma^*$ be a social optimum. We first prove that $\scost(\sigma) \leq 2 \scost(\sigma^*)$.
Let $B_1,\ldots, B_{k+1}$ be the $k+1$ blocks induced by $\sigma$. For each block $B_i$, let $m_i$ (resp., $M_i$) denote the minimum (resp., maximum) type of an agent that occupies a node of block $B_i$. By~\autoref{lemma:upper_bound_sum}, we have that
\begin{equation}\label{eq:upper_bound_p_G}
\scost(\sigma) \leq 2\sum_{i=1}^{k+1}(M_i-m_i).
\end{equation}

Let $B^*_1,\ldots, B^*_h$, be the $h$ blocks of $\sigma^*$. For each block $B_i^*$, let $m_i^*$ (resp., $M_i^*$) denote the minimum (resp., maximum) type of an agent that occupies a node of block $B_i^*$. By~\autoref{lemma:lower_bound_sum} we have that
\begin{equation}\label{eq:upper_bound_p_prime_G}
\sum_{i=1}^{h}(M^*_i-m^*_i) \leq \scost(\sigma^*).
\end{equation}
We conclude the proof by showing that
\begin{equation}\label{eq:upper_bound_intervals}
\sum_{i=1}^{k+1}(M_i-m_i) \leq \sum_{i=1}^h(M^*_i-m^*_i).
\end{equation}
In fact, by combining~\autoref{eq:upper_bound_p_G} with~\autoref{eq:upper_bound_p_prime_G} and~\autoref{eq:upper_bound_intervals}, we obtain
$\scost(\sigma) \leq 2\scost(\sigma^*)$
as desired.

To prove~\autoref{eq:upper_bound_intervals}, it is convenient to see the blocks as intervals on the line, where the $i$-th block of $\sigma$ (resp., $\sigma^*$) is associated with the interval $[m_i,M_i]$ (resp., $[m^*_i,M^*_i]$). Observe that for each real value $t \in [m_i,M_i]$ (resp., $t \in [m^*_i,M^*_i]$) there are two consecutive agents with types $t'$ and $t''$ in $B_i$ (resp., $B^*_i$) such that $t$ is \emph{covered} by the interval defined by $t'$ and $t''$, i.e., $\min\{t',t''\} \leq t \leq \max\{t',t''\}$. The $k+1$ intervals $\{[m_i,M_i] \mid i \in[k+1]\}$ are pairwise disjoint. Moreover, they leave uncovered all the values of the largest $k$ intervals of the form $(t(j),t(j+1))$, where $j \in [n-1]$. Similarly, the $h$ intervals $\{[m^*_i,M^*_i] \mid i \in [h]\}$ leave uncovered at most $h \leq k$ intervals of the form $(t(j),t(j+1))$, where $j \in [n-1]$. Therefore, $\sum_{i=1}^{k+1}(M_i-m_i) \leq \sum_{i=1}^h (M^*_i - m^*_i)$, i.e.,~\autoref{eq:upper_bound_intervals} holds.
\end{proof}

    \subsubsection{Average Type-Distance Game}
		\begin{theorem} \label{thm:Swap_ADG_PoS}
			The $\hpos$ for the S-ADG on $\Delta$-reg. graphs is $1$.
		\end{theorem}

		\begin{proof}
			Let $\sigma$ be a strategy profile in a $\Delta$-regular graph $G$. We have already
			shown in~\autoref{thm:Swap_ADG_Ex} that \[\Phi(\sigma) = \sum_{\{u,v\}\in E}
			d(\sigma^{-1}(u), \sigma^{-1}(v)) = 2 \cdot \scost(\sigma)\] is an ordinal
			potential function for the S-ADG. Thus, the social optimum $\sigma^*$ is also a global optimum for $\Phi$. It follows that $\sigma^*$ is an equilibrium and so the $\hpos$ is $1$.
		\end{proof}
The next lemma is used to show the existence of a SE for the S-ADG that is also a social optimum.

\begin{lemma}\label{lm:S_ADG_path_lower_bound}
Let $\sigma$ be a strategy profile for an S-ADG played on a path. Then $\scost(\sigma) \geq \frac{1}{2}(t(2)-t(1))+\frac{1}{2}(t(n)-t(n-1))+(t(n)-t(1))$.
\end{lemma}
\begin{proof}
Let $i$ and $j$ be the indices of the nodes where agents 1 and $n$ have been placed in $\sigma$, respectively. W.l.o.g., we assume that $i < j$.
Let $\ell_1,\ldots,\ell_k$ denote the indices of the agents that are placed, in order, in the nodes of the path, starting from node $i$ and ending in node $j$. So, agent $\ell_s$, with $s \in [k]$, occupies node $i+s-1$. We denote by $t_s=t(\ell_s)$ the type of agent $\ell_s$.
As the distance of agent 1 from any other agent is at least $t(2)-t(1)$, we can lower bound the cost of agent 1 by the following.
\begin{equation}\label{eq:cost_1}
\scost_1(\sigma)\geq \frac{1}{2}(t_2-t_1)+\frac{1}{2}(t(2)-t(1)).
\end{equation}
Similarly, as the distance of agent 1 from any other agent is at least $t(2)-t(1)$, we can lower bound the cost of agent 1 by the following.
\begin{equation}\label{eq:cost_n}
\scost_n(\sigma) \geq \frac{1}{2}(t_n-t_{n-1})+\frac{1}{2}(t(n)-t(n-1)).
\end{equation}
Thanks to~\autoref{eq:cost_1} and~\autoref{eq:cost_n}, we have
\begin{eqnarray*}
\sum_{s=1}^k \scost_s(\sigma) & \geq & \frac{1}{2}(t_2-t_1)+\frac{1}{2}(t(2)-t(1))\\
& & \, +\frac{1}{2}(t_n-t_{n-1})+\frac{1}{2}(t(n)-t(n-1)) \\
& & \, + \frac{1}{2}\sum_{s=2}^{k-1}\big((t_{s}-t_{s-1})+(t_{s+1}-t_{s})\big)\\
& = & \frac{1}{2}(t_2-t_1)+\frac{1}{2}(t(2)-t(1)) + \\
& & \, +\frac{1}{2}(t_n-t_{n-1})+\frac{1}{2}(t(n)-t(n-1)) \\
& & \, +\sum_{s=1}^{k-1}(t_{s+1}-t_{s}) \\
& & \, - \frac{1}{2}(t_2-t_1)-\frac{1}{2}(t_k-t_{k-1})\\
& = & \frac{1}{2}(t(2)-t(1))+\frac{1}{2}(t(n)-t(n-1))\\
& & \, + (t(n)-t(1)). \hspace*{3cm}\qedhere
\end{eqnarray*}
\end{proof}

\begin{theorem}\label{PoS-AVG-path}
The $\hpos$ for the S-ADG on paths is 1.
\end{theorem}
\begin{proof}
Let $\sigma$ be the strategy profile in which agent $i$ is placed on the $i$-th node of the path, where the nodes of the path are numbered from 1 to $n$ according to their left-to-right ordering.

We prove that $\sigma$ is a SE. Let $i$ and $j$, with $i < j$ be any two agents. Observe that the cost of agent 1 (resp., $n$) in any strategy profile $\sigma'$ is lower bounded by $t(2)-t(1)$ (resp., $t(n)-t(n-1)$). Hence neither $1$ nor $n$ wants to swap in $\sigma$. Therefore, we only need to consider the case in which $i > 1$ and $j < n$. It is easy to observe that at least one of the two agents $i$ and $j$ does not want to swap.

Let $\sigma'$ be any strategy profile.
By~\autoref{lm:S_ADG_path_lower_bound}, we have
\begin{eqnarray*}
\scost(\sigma) & = & (t(2)-t(1))+(t(n)-t(n-1))\\
& & \,  +\frac{1}{2}\sum_{i=2}^{n-1} \big((t(i)-t(i-1))+(t(i+1)-t(i))\big)\\
& = & \frac{1}{2}(t(2)-t(1))+\frac{1}{2}(t(n)-t(n-1))\\
& & \,+ (t(n)-t(1))\\
& \geq & \scost(\sigma'). \hspace*{4.5cm}\qedhere
\end{eqnarray*}
\end{proof}

    \subsubsection{Cutoff Game}

	\begin{theorem} \label{thm:Swap_CG_PoS}
		The $\hpos$ for the S-CG on $\Delta$-reg. graphs is~$1$.
	\end{theorem}

	\begin{proof}
		By~\autoref{thm:Swap_CG_Ex_almostreg}, we have shown that \[\Phi(\sigma) = \lvert \{u,v\}\in E \colon \sigma^{-1}(u)\in N_{\sigma^{-1}(v)}^+(\sigma)\}\lvert\] is an ordinal potential function in S-CGs.

        On the other hand, we have
        \begin{eqnarray*}
        \scost(\sigma) & = & \sum_{i\in [n]}\cost i(\sigma)\\
        & = & \sum_{i\in [n]}\frac{|N_i(\sigma)|-|N^+_i(\sigma)|}{|N_i(\sigma)|}\\
        & = & n-\frac{1}{\Delta}\sum_{i\in [n]}|N^+_i(\sigma)|\\
        & = & n-\frac{2}{\Delta}\Phi(\sigma).
        \end{eqnarray*}
        So, any global maximum of  $\Phi$, which corresponds to a SE, is also a global minimum of the social cost.
	\end{proof}

\begin{theorem}\label{thm:Swap_CP_PoS_path}
The $\hpos$ for the S-CG on paths is at most 2.
\end{theorem}
\begin{proof}
Let $\sigma$ be the strategy profile in which agent $i$ is placed on the $i$-th node of the path, where the nodes of the path are numbered from 1 to $n$ according to their left-to-right ordering. One can easily observe that $\sigma$ is a SE.

Consider the complete graph on the agents, where the weight of the edge between agent $i$ and agent $j$ is equal to $d(i,j)=|t_i-t_j|$. We observe that the path associated with any strategy profile $\sigma'$ is a Hamiltonian path, say $H'$, of the weighted graph. Moreover, the path associated with the strategy profile $\sigma$ is also a minimum spanning tree, say $H$, of our weighted graph. Let $k$ be the number of edges of $H$ whose weights are above the cutoff threshold of $\lambda$. By the property of minimum spanning trees, we have that $H'$ contains at least $k$ edges whose weights are above the cutoff threshold of $\lambda$.

This implies that the cost of $\sigma'$ can be lower bounded by $k$ as each of the $k$ expensive edges contributes by at least $1$ in $\scost(\sigma')$ (more precisely, there are two occurrences of the edge, one for each of its end-nodes, and each occurrence contributes to the cost of the end-node with at least $1/2$).

In the following, we only consider the case in which $k\geq 2$. In fact, when $k=0$, we have that $\scost(\sigma)=0$, so $\sigma$ is an optimal strategy profile. Furthermore, when $k=1$, we have that $3/2$ is an upper bound on the cost of $\sigma$, unless $n=2$, in which case $\sigma$ would be again an optimal strategy profile.

The cost of $\sigma$ can be upper bounded by $k+2$, where the bound is met when the two end-nodes of the path both have a cost of $1$. As $\sigma'$ is any strategy profile and $k\geq 2$, the $\hpos$ can be upper bounded by the following
\[
\frac{\scost(\sigma)}{\scost(\sigma')} \leq \frac{k+2}{k} \leq 2.
\qedhere\]
\end{proof}

\section{Omitted Details from~\Cref{sec:experiments}}\label{apx:experiments}
For the swap variant of the CG, we additionally investigated the influence of the density of the underlying $\Delta$-regular graph on the obtained equilibria. For this,~\autoref{fig:equicutDelta} shows equilibria for different values of $\Delta$.
\begin{figure}[h]
    \centering
    \begin{subfigure}{0.48\linewidth}
        \centering
        \includegraphics[width=0.7\textwidth]{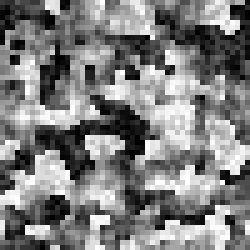}
        \caption{4-regular toroidal grid\label{fig:4reg}}
    \end{subfigure}
    \hfil
    \begin{subfigure}{0.48\linewidth}
        \centering
        \includegraphics[width=0.7\textwidth]{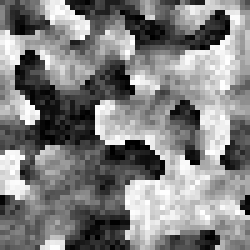}
        \caption{8-regular toroidal grid\label{fig:8reg}}
    \end{subfigure}

    \medskip

    \begin{subfigure}{0.48\linewidth}
        \centering
        \includegraphics[width=0.7\textwidth]{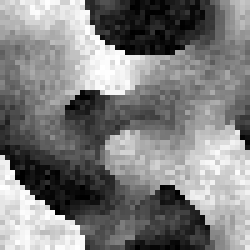}
        \caption{24-regular toroidal grid\label{fig:24reg}}
    \end{subfigure}
    \hfil
    \begin{subfigure}{0.48\linewidth}
        \centering
        \includegraphics[width=0.7\textwidth]{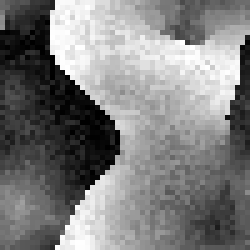}
        \caption{48-regular toroidal grid\label{fig:48reg}}
    \end{subfigure}
    \caption{SE in the S-CG with cutoff $\lambda = 0.2$ in regular toroidal grid
    graphs with a different degree. In \textbf{(a)}, every pixel is connected to
other pixels sharing an edge. For graphs \textbf{(b)}, \textbf{(c)}, and
\textbf{(d)}, every pixel is connected to every other pixel within a square with
side length 3, 5, and 7 around itself. All graphs live on a torus.\label{fig:equicutDelta}}
\end{figure}
We find that with growing $\Delta$ also the segregation of the equilibria increases.

The influence on the cutoff value $\lambda$ on the obtained swap equilibria is depicted in~\autoref{fig:lambda}.
\begin{figure}[h]
    \centering
    \begin{subfigure}{0.48\linewidth}
        \centering
        \includegraphics[width=0.7\textwidth]{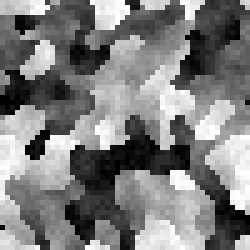}
        \caption{$\lambda = 0.1$\label{fig:0.1}}
    \end{subfigure}
    \hfil
    \begin{subfigure}{0.48\linewidth}
        \centering
        \includegraphics[width=0.7\textwidth]{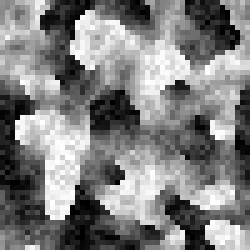}
        \caption{$\lambda = 0.2$\label{fig:0.2}}
    \end{subfigure}

    \medskip

    \begin{subfigure}{0.48\linewidth}
        \centering
        \includegraphics[width=0.7\textwidth]{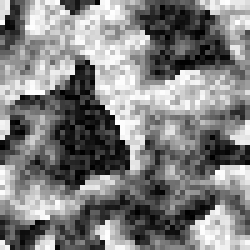}
        \caption{$\lambda = 0.3$\label{fig:0.3}}
    \end{subfigure}
    \hfil
    \begin{subfigure}{0.48\linewidth}
        \centering
        \includegraphics[width=0.7\textwidth]{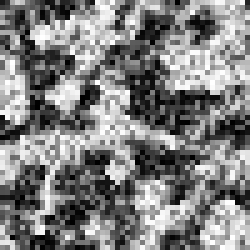}
        \caption{$\lambda = 0.5$\label{fig:0.5}}
    \end{subfigure}
    \caption{\label{fig:lambda}Swap Equilibria for different $\lambda$-values in the
    CG.}
\end{figure}
For this, we see the expected result that with increasing cutoff value $\lambda$, i.e., if agents get more and more tolerant of who they accept as friends, the obtained segregation decreases.

\end{document}